\documentclass[pdflatex,sn-mathphys-num]{sn-jnl}

\usepackage{graphicx,tikz-cd}%
\usetikzlibrary{backgrounds}
\usetikzlibrary{decorations, decorations.pathmorphing, decorations.markings}
\usepackage{multirow}%
\usepackage{amsmath,amssymb,amsfonts}%
\usepackage{amsthm,mathtools}%
\usepackage{mathrsfs}%
\usepackage[title]{appendix}%
\usepackage{mathrsfs} 
\usepackage{xcolor}%
\usepackage{textcomp}%
\usepackage{manyfoot}%
\usepackage{booktabs}%
\usepackage{algorithm}%
\usepackage{algorithmicx}%
\usepackage{algpseudocode}%
\usepackage{listings}%
\usepackage{mathpazo}
\usepackage{dutchcal}

\theoremstyle{thmstyleone}%
\hypersetup{
	colorlinks=ture,
	linkcolor=blue,
	filecolor=gray,
	urlcolor=blue,
	citecolor=blue}
\numberwithin{equation}{section}
\newtheorem{definition}{Definition}[section]
\newtheorem{example}{Example}[section]
\newtheorem{theorem}{Theorem}[section]
\newtheorem{lemma}{Lemma}[section]	
\newtheorem{proposition}{Proposition}[section]

\allowdisplaybreaks[4] 

\begin{document}

\title[Higher Chern--Simons Theory]{Higher Chern--Simons Theory in $2n+2$ Dimensions for Balanced 2-term $L_\infty$-Algebras}


\author[1]{\fnm{Danhua} \sur{Song}}\email{songdh@pku.edu.cn}

\author*[2]{\fnm{Yibo} \sur{Wang}}\email{wangyb@amss.ac.cn}


\affil[1]{\orgdiv{Beijing International Center for Mathematical Research}, \orgname{Peking University}, \orgaddress{ \city{Beijing}, \postcode{100871}, \country{China}}}


\affil*[2]{\orgdiv{Academy of Mathematics and Systems Science}, \orgname{Chinese Academy of Sciences}, \orgaddress{\city{Beijing}, \postcode{100190}, \country{China}}}


\abstract{We construct a semistrict higher Chern--Simons (HCS) gauge theory in $2n+2$ dimensions associated with balanced 2-term $L_\infty$-algebras. Starting from the homotopy Maurer--Cartan theory, we first introduce 2-term $L_\infty$-algebra gauge theory, and show that there is a four-dimensional HCS construction. Then we extend invariant bilinear pairings to invariant multilinear forms of the appropriate degree, and define a $(2n+3)$-dimensional higher Pontryagin--Chern form, which is closed and invariant under infinitesimal gauge transformations. Its transgression yields an explicit $(2n+2)$-dimensional HCS form. We further establish a higher Chern--Weil theorem that generates higher transgression forms, and prove that the HCS theory is a distinguished instance of the higher transgression gauge theory. Finally, we apply the extended Cartan homotopy formula in this semistrict setting, and show that it is a common origin of both the higher Chern--Weil theorem and the associated triangle equation.  }




\maketitle
\tableofcontents

\section{Introduction}
Chern--Simons (CS) theory provides a distinguished example of a topological gauge theory. Higher Chern--Simons (HCS) theory extends this framework to higher gauge theory, in which the gauge algebra is replaced by a higher algebraic structure, such as an  $L_\infty$-algebra \cite{Fiorenza-Rogers-Schreiber,PRCS-2016,RZ-AKSZ}. Homotopy Maurer--Cartan (MC) theory yields a natural $L_\infty$-algebraic formulation of the HCS theory and, in particular, associates a four-dimensional HCS model with a 2-term $L_\infty$-algebra \cite{GR,BLCM}. While higher-dimensional extensions have been constructed in the strict case \cite{DHS-4,DHS-5}, a corresponding construction for general semistrict 2-term $L_\infty$-algebras has remained unavailable.
In this work, we construct a family of semistrict HCS theories in even spacetime dimensions $2n+2$, associated with balanced 2-term $L_\infty$-algebras. Our construction is based on the homotopy MC framework and on a transgression mechanism inspired by extended CS theory \cite{IAGS,FIP,FIPSPSS}. It extends the known four-dimensional semistrict HCS theories to arbitrary even dimensions and reduces, in the strict limit, to the higher-dimensional strict HCS theories of \cite{DHS-4,DHS-5}.

Higher gauge theory generalizes ordinary gauge theory by replacing Lie groups, Lie algebras, and principal bundles with higher Lie groups, $L_\infty$-algebras, and principal higher bundles \cite{Girelli-Pfeiffer,Baez-2007, Baez.2010}. Locally, a higher connection consists of differential forms of several degrees, together with corresponding higher curvatures. Such structures naturally couple to extended objects and arise in a variety of contexts in mathematical physics, including string and brane theories, supergravity, and higher-form gauge theories \cite{Cagnacci-19,Polchinski,Becker-2007,Zeitlin-09,Baez-543, CR-455}. Since the ordinary CS action is formulated in terms of a connection and an invariant polynomial on a Lie algebra, it is natural to seek its extension within this higher gauge-theoretic setting.
The relevant algebraic framework is provided by $L_\infty$-algebras, namely, graded vector spaces equipped with multilinear brackets satisfying the homotopy Jacobi identities \cite{MSJS,TLJS,MP,Kraft-Schnitzer}. Homotopy MC theories associated with cyclic $L_\infty$-algebras provide a broad class of gauge theories with action functionals \cite{B. Zwiebach-1993}. In the case of a Lie algebra tensored with the de Rham complex of an oriented three-manifold, the corresponding homotopy MC action reproduces ordinary CS theory. More generally, the standard homotopy MC construction gives rise to HCS theories associated with $k$-term $L_\infty$-algebras in dimension $k+ 2$ \cite{GR,BLCM}. In particular, a 2-term $L_\infty$-algebra, equivalently a semistrict Lie 2-algebra \cite{JCB-ASC}, yields a four-dimensional HCS theory.

Several four-dimensional semistrict HCS models have been developed. Zucchini formulated a semistrict HCS theory using the Alexandrov--Kontsevich--Schwartz--Zaboronsky formalism \cite{RZ-AKSZ}, while Soncini and Zucchini constructed a model based on a semistrict Lie 2-algebra equipped with an invariant non-degenerate bilinear form \cite{Zucchini-2014-1}. A further distinguished example, associated with a skeletal semistrict Lie 2-algebra, was identified in \cite{Zucchini-2016}. In parallel, strict models based on crossed modules and strict Lie 2-groups have been studied extensively, including their holonomy invariants and applications to integrable field theory \cite{Zucchini-2021,Zucchini-I,Zucchini-II,Zucchini-2019,Schenkel-Vicedo,HC-2024}. Higher-dimensional constructions have also been obtained in the strict setting, including a five-dimensional theory associated with a strict Lie 3-group \cite{DHS-JHEP}. These developments leave open the problem of constructing higher-dimensional HCS theories for general semistrict 2-term $L_\infty$-algebras.

A useful starting point is provided by extended CS theory and its formulation in terms of free differential algebras (FDA) \cite{Salgado, SS, PSSS}. In this setting, one considers closed gauge-invariant forms $\Gamma_{2n+p}$, with $p=3,4,6,8$, generalizing the ordinary Pontryagin--Chern forms \cite{GS,GS1,SKGS}. The corresponding extended CS forms $\mathfrak{C}^{(2n+p-1)}$ are defined locally  by
\begin{equation*}
	d\mathfrak{C}^{(2n+p-1)}=\Gamma_{2n+p}.
\end{equation*}
A particularly relevant example is
\begin{equation}\label{ei}
	\Gamma_{2n+3}=\langle F^n H\rangle,
\end{equation}
where $\langle\cdots\rangle$ is a symmetric invariant $(n+1)$-linear polynomial on the Lie algebra $\mathcal{g}$ of the underlying gauge group:
\begin{equation*}
	\langle\cdots\rangle:\mathcal{g}^{\,n+1}\to\mathbb{R}.
\end{equation*}
Here $F=dA+A\wedge A$ is the curvature of a $\mathcal{g}$-valued 1-form $A$, and $H=dB+[A,B]$ is the curvature of a $\mathcal{g}$-valued 2-form.  The associated  $(2n+2)$-form  $\mathfrak{C}^{2n+2}$ satisfies
$ d\mathfrak{C}^{2n+2}= \Gamma_{2n+3}$.
The generalized Chern--Weil theorem and the associated triangle equation can be derived by means of the extended Cartan homotopy formula (ECHF) \cite{Zumino,PRRJ,FEP}.
Song \textit{et al.}\ extended this construction to strict Lie 2-algebras \cite{DHS-4}. In their formulation, the pair $(A, B)$ is promoted to a Lie 2-algebra-valued 2-connection, and $(F, H)$ is replaced by the corresponding 2-curvature. The invariant polynomial in \eqref{ei} is replaced by a multilinear invariant form on the Lie 2-algebra, leading to a higher Pontryagin--Chern form and an associated $(2n+2)$-dimensional strict HCS form. Higher analogues of the Chern--Weil and transgression formulas, together with the corresponding triangle equation, were subsequently obtained by applying the ECHF \cite{DHS-5}. The construction developed in the present paper extends this strict framework to balanced 2-term $L_\infty$-algebras.

More precisely, we introduce a closed $(2n+3)$-form $\Gamma_{2n+3}$  associated with a balanced  $2$-term $L_{\infty}$-algebra, and construct a local $(2n+2)$-dimensional semistrict HCS form $\mathfrak{C}^{2n+2}$ satisfying the corresponding transgression relation. We derive higher Chern--Weil transgression formula and study its local gauge-theoretic properties. By applying the ECHF, we establish semistrict analogues of the generalized Chern--Weil theorem and the triangle equation. The higher-dimensional theories considered here should be distinguished from the standard $(k+2)$-dimensional homotopy MC construction: they arise from a transgression construction involving invariant multilinear pairings and the 2-curvature valued in a balanced $2$-term $L_{\infty}$-algebra. In the strict limit, our results recover the $(2n+2)$-dimensional  strict HCS theory of \cite{DHS-4,DHS-5}; for $n=1$, they reproduce the four-dimensional semistrict HCS models of \cite{RZ-AKSZ,Zucchini-2014-1}. Throughout this paper, we work locally and do not address global aspects of the resulting higher gauge theories.

The paper is organized as follows. Section~\ref{s1} reviews the theory of $L_{\infty}$-algebras equipped with cyclic invariant pairings, which constitutes the central algebraic framework for the development of HCS theory. Within this setting, we introduce the homotopy MC formalism, which provides a natural $L_{\infty}$-algebra extension of the ordinary CS theory.
In Section~\ref{sec:2term}, we specialize the homotopy MC  theory to balanced 2-term $L_{\infty}$-algebras. We explicitly describe the associated higher gauge fields, curvature forms, and their infinitesimal gauge transformations. As a direct application, we recover the four-dimensional HCS theory within this framework.        
Section~\ref{sec:2n+2-hCS} develops the $(2n+2)$-dimensional HCS theory. We first extend the cyclic invariant pairing to an $(n+1)$-linear form and extend it to the space of algebra-valued differential forms, establishing that the extended pairing inherits the original graded symmetry and cyclicity. We then construct the higher Pontryagin--Chern form in $2n+3$ dimensions, proving its closure  and gauge invariance, and we give the explicit formula for $(2n+2)$-dimensional HCS form by the infinitesimal transgression formula. Finally, we establish the higher Chern--Weil theorem, introduce the associated higher transgression form, clarify its precise relation to the HCS form, and study the higher gauge theory it defines.        
In Section~\ref{6-s}, we prove that the ECHF is compatible with semistrict HCS theory, and moreover constitutes the common origin of both the higher Chern--Weil theorem and the higher triangle equation.        
We finish in Section~\ref{7-s} with final conclusions and some considerations on future possible developments.

\section{$L_\infty$-Algebras and Homotopy Maurer--Cartan Theory}\label{s1}
This section recalls the $L_\infty$-algebraic framework for homotopy MC theory and its application to HCS theory. Further background on these constructions can be found in \cite{BLCM,BTLCM}.

\subsection{Cyclic $L_{\infty}$-algebras}
Throughout, all vector spaces are assumed to be $\mathbb Z$-graded, and all signs are determined by the Koszul rule  \cite{ASCFS,NW} . We first recall the definition of an $L_\infty$-algebra and then introduce cyclic invariant pairings.
\begin{definition}\label{infty}
	An \emph{$L_{\infty}$-algebra} is a $\mathbb{Z}$-graded vector space $\mathfrak{L}=\bigoplus_{k\in\mathbb{Z}}\mathfrak{L}_k$ equipped with a collection of graded, totally antisymmetric multilinear maps
	\[
	\mu_i:\underbrace{\mathfrak{L}\times\cdots\times\mathfrak{L}}_{i\text{ copies}}\longrightarrow\mathfrak{L},\qquad i\in\mathbb{N}^+,
	\]
	of degree $\lvert\mu_i\rvert=2-i$, which is called the \emph{higher products}. These maps satisfy the \emph{higher (homotopy) Jacobi identities}: for every $i\geq 1$ and all homogeneous elements $l_1,\dots,l_i\in\mathfrak{L}$,
	\begin{equation}\label{HJI}
		\sum_{r+s=i}\ \sum_{\sigma\in\mathrm{Sh}(r,s)}
		(-1)^{s}\,\chi(\sigma; l_1,\dots,l_{i})\,
		\mu_{s+1}\!\big(\mu_r(l_{\sigma(1)},\dots,l_{\sigma(r)}),
		l_{\sigma(r+1)},\dots,l_{\sigma(i)}\big)=0.
	\end{equation}
	Here $\mathrm{Sh}(r,s)$ denotes the set of \emph{$(r,s)$-unshuffles}, in which permutations $\sigma$ of $\{1,\dots,i\}$ satisfy $\sigma(1)<\cdots<\sigma(r)$ and $\sigma(r+1)<\cdots<\sigma(i)$, and $\chi(\sigma;l_1,\dots,l_i)$ is the Koszul sign determined by 
\begin{equation*}
	l_1\wedge\cdots\wedge l_i
	=\chi(\sigma;l_1,\dots,l_i)\,
	l_{\sigma(1)}\wedge\cdots\wedge l_{\sigma(i)}.
	\end{equation*}
\end{definition}

For a homogeneous element $l\in\mathfrak L$, we denote its degree by $|l|$. The graded antisymmetry of the higher brackets is expressed by
\begin{equation*}
	\mu_i(\cdots, l_a, l_{a+1},\cdots)=-(-1)^{|l_a| |l_{a+1}|} 	\mu_i(\cdots, l_{a+1}, l_a\cdots)
\end{equation*}
for all homogeneous elements $l_a,l_{a+1}\in\mathfrak L$.
The first three homotopy Jacobi identities take the following form. Let $l_1,l_2,l_3\in\mathfrak L$ be homogeneous:
\begin{itemize}
	\item For $i=1$, the identity \eqref{HJI} reduces to 
	\begin{equation*}
		\mu_1(\mu_1(l_1))=0.
	\end{equation*}
	Equivalently, $\mu_1^2=0$. Hence $(\mathfrak L,\mu_1)$ is a cochain complex, with $\mu_1$ a differential of degree $1$.
	\item For $i=2$, the identity \eqref{HJI} becomes
	\begin{equation*}
		\mu_1(\mu_2(l_1,l_2))
		-\mu_2(\mu_1(l_1),l_2)-(-1)^{|l_1|}\,\mu_2(l_1,\mu_1(l_2))=0.
	\end{equation*}
	Thus, $\mu_1$ is a graded derivation of degree $1$ with respect to the binary bracket $\mu_2$. 
	\item For $i=3$,  the identity \eqref{HJI} reads 
	\begin{align*}
		&\mu_2(\mu_2(l_1,l_2),l_3)
		+(-1)^{(|l_1|+|l_2|)\,|l_3|}\,\mu_2(\mu_2(l_3,l_1),l_2)
		+(-1)^{(|l_2|+|l_3|)\,|l_1|}\,\mu_2(\mu_2(l_2,l_3),l_1) \\
		&+\,\mu_1(\mu_3(l_1,l_2,l_3))
		+\mu_3(\mu_1(l_1),l_2,l_3)
		+(-1)^{|l_1|}\,\mu_3(l_1,\mu_1(l_2),l_3)\\
		&
		+(-1)^{|l_1|+|l_2|}\,\mu_3(l_1,l_2,\mu_1(l_3))=0.
	\end{align*}
	Thus, in an $L_\infty$-algebra the Jacobi identity for  $\mu_2$ holds only up to homotopy, with the failure controlled by the differential $\mu_1$ and the trilinear bracket $\mu_3$.
\end{itemize}

An $L_\infty$-algebra is called \emph{minimal} if $\mu_1=0$, \emph{strict} if $\mu_i=0$ for all $i>2$, and \emph{abelian} if $\mu_i=0$ for all $i>1$. 
An $L_\infty$-algebra is called a \emph{$k$-term $L_\infty$-algebra}, or equivalently a \emph{Lie $k$-algebra}, if it is concentrated in degrees $-k+1,\dots,0$, namely,
\begin{equation*}
	\mathfrak L=\mathfrak L_{-k+1}\oplus\cdots\oplus \mathfrak L_0.
\end{equation*}
In particular, suppose that $\mathfrak L$ is concentrated in degree $0$, so that $\mathfrak L=\mathfrak L_0$.
Since $\mu_i$ has degree $2-i$, degree considerations imply that $\mu_i$ can be non-vanishing only for $i=2$. Hence $\mu_1=0$ and $\mu_i=0$ for all $i\geq 3$. The graded antisymmetry of $\mu_2$  then reduces to ordinary antisymmetry, while the $L_\infty$-identity for $i=3$ becomes the ordinary Jacobi identity. Therefore,  $1$-term $L_\infty$-algebras are precisely ordinary Lie algebras.

Cyclic $L_\infty$-algebras provide the algebraic input for the construction of HCS theories \cite{OZ}.
In this setting, an invariant pairing plays a role analogous to that of an invariant trace in ordinary
gauge theory. We therefore recall the notion of a cyclic structure on an $L_\infty$-algebra. 
\begin{definition}
	Let $(\mathfrak{L},\{\mu_i\}_{i\ge 1})$ be an $L_\infty$-algebra. 
	A \emph{pairing of degree $q$} on $\mathfrak{L}$ is a non-degenerate graded symmetric bilinear form 
	\begin{equation*}
		\langle -,-\rangle:\mathfrak{L}\times\mathfrak{L}\to\mathbb{R},
	\end{equation*}
	where $\mathbb{R}$ is regarded as a graded vector space concentrated in degree 0. It satisfies,
	\begin{equation*}
		\langle l, l'\rangle=0, \qquad \text{unless} \qquad |l|+|l'|+q=0
	\end{equation*}
	for homogeneous elements $l, l'\in\mathfrak{L}$.
	The pairing is called \emph{invariant} if, for every $i\geq 1$ and all homogeneous elements $l_1,\dots,l_{i+1}\in\mathfrak{L}$,
	\begin{equation}\label{cyc}
		\langle l_1, \mu_i(l_2, \dots, l_{i+1})\rangle=(-1)^{i+i(|l_1|+|l_{i+1}|)+
			|l_{i+1}|(|l_1|+\cdots+|l_i|)}\langle l_{i+1}, \mu_i(l_1, \dots, l_i)\rangle.
	\end{equation}
	An $L_\infty$-algebra equipped with an invariant pairing is called a  \emph{cyclic $L_\infty$-algebra}.
\end{definition}

The graded symmetry of $\langle -,-\rangle$ means that
\begin{equation*}
	\langle l_1, l_2\rangle=(-1)^{|l_1| |l_2|}\langle l_2, l_1\rangle
\end{equation*}
for all homogeneous elements $l_1, l_2\in\mathfrak{L}$.
Condition~\eqref{cyc} expresses the graded cyclic invariance of the
pairing with respect to the higher brackets. Such cyclic structures were introduced in the foundational works of Penkava and Kontsevich; see, for example, \cite{MP,Kontsevich-92}. If $\mathfrak L=\mathfrak g$ is an ordinary Lie algebra concentrated in degree $0$, with $\mu_2=[-,-]$, then \eqref{cyc} reduces to the usual invariance condition
\begin{equation*}
	\langle [X_1,X_2],X_3\rangle=-\langle X_2,[X_1,X_3]\rangle,
	\qquad \forall X_1,X_2,X_3\in\mathfrak{g}.
\end{equation*}

We next recall the $L_\infty$-algebra of $\mathfrak{L}$-valued differential forms. Let $M$ be a smooth manifold. Since the de Rham complex $(\Omega^\bullet(M), d, \wedge)$ is a differential graded commutative algebra, the graded vector space
\begin{equation*}
	\Omega^\bullet(M, \mathfrak{L})\coloneqq \bigoplus_{k\in \mathbb{Z}}\Omega^\bullet_k(M, \mathfrak{L}) 
	\quad 
	\text{with}
	\quad 
	\Omega^\bullet_k(M, \mathfrak{L}) \coloneqq \bigoplus_{i+j=k}\Omega^i(M)\otimes \mathfrak{L}_j
\end{equation*}
carries a natural $L_\infty$-algebra structure \cite{BJCSMW,BLCM}.
For homogeneous elements $\omega\in\Omega^\bullet(M)$ and $ l \in\mathfrak{L}$, we write
\begin{equation*}
	|\omega\otimes l |=|\omega|+| l |,
\end{equation*}
where $|\omega|$ denotes the form degree of $\omega$.
The induced brackets $\{\mu_i^*\}_{i\ge 1}$ have degree $2-i$.  Their actions on homogeneous elements are given by
\begin{align}
	\mu_1^\ast(\omega\otimes l)
	&=
	d\omega\otimes l
	+
	(-1)^{|\omega|}\omega\otimes\mu_1(l),
	\label{mu1}\\
	\mu_i^\ast(\omega_1\otimes l_1,\ldots,\omega_i\otimes l_i)
	&=
	(-1)^{\varepsilon_i}
	(\omega_1\wedge\cdots\wedge\omega_i)
	\otimes\mu_i(l_1,\ldots,l_i),
	\qquad i\geq2,
	\label{mun}
\end{align}
where
\begin{equation*}
	\varepsilon_i=i\sum_{j=1}^i|\omega_j|+\sum_{j=0}^{i-2}|\omega_{i-j}|\sum_{k=1}^{i-j-1}| l _k|.
\end{equation*}
Equivalently, for $i\geq 2$,  the bracket $\mu_i^*$ is induced by $m_i\otimes \mu_i$, where
\begin{equation*}
	m_i(\omega_1,\dots,\omega_i)=\omega_1\wedge\cdots\wedge \omega_i,
\end{equation*}
with the Koszul signs prescribed by the graded tensor-product
convention. The higher Jacobi identities for $\{\mu_i^*\}_{i\ge 1}$ follow from those of $\{\mu_i\}_{i\ge 1}$ together with the differential
graded-commutative algebra structure of $\Omega^\bullet(M)$; see
\cite{BLCM}.

Suppose now that $\mathfrak L$ is cyclic and that  $M$ is a compact,
oriented manifold without boundary. The invariant pairing on $\mathfrak L$ induces a natural pairing on $\Omega^\bullet(M, \mathfrak{L})$ by
\begin{equation}\label{pair}
	\ll \omega_1\otimes l_1,\;\omega_2\otimes l_2\gg
	\coloneqq (-1)^{|\omega_2|\,|l_1|}\int_M \omega_1\wedge \omega_2\,\langle l_1,l_2\rangle,
\end{equation}
for homogeneous $\omega_{i}\in\Omega^\bullet(M)$ and $l_{i}\in\mathfrak L$, $i=1, 2$, and by linear extension otherwise. The
pairing in \eqref{pair} vanishes unless the form-degree component of $\omega_1\wedge \omega_2$ is equal to $\text{dim} (M)$. With the degree
convention above, it has degree $q-\text{dim} (M)$.

The following result is the generalization of \eqref{cyc} to the case of $L_{\infty}$-valued differential forms.
\begin{proposition}\label{lem:cyclic-induced}
	Let $(\mathfrak L,\{\mu_i\}_{i\ge 1},\langle-,-\rangle)$ be a cyclic $L_\infty$-algebra, and $M$ be a compact, oriented manifold without boundary. Then the pairing \eqref{pair} on
	$\Omega^\bullet(M,\mathfrak L)$  is cyclic with respect to the induced
	brackets $\{\mu_i^*\}_{i\ge 1}$.  More precisely, for every $i\ge 1$ and all homogeneous elements $a_1,\dots,a_{i+1}\in\Omega^\bullet(M,\mathfrak L)$,
	\begin{equation}\label{eq:cyc-induced}
		\ll a_1,\mu_i^*(a_2,\dots,a_{i+1})\gg
		=
		(-1)^{\xi_i(a_1,\dots,a_{i+1})}\,
		\ll a_{i+1},\mu_i^*(a_1,\dots,a_i)\gg,
	\end{equation}
	where 
	\begin{equation*}\label{eq:Sigmai}
		\xi_i(a_1,\dots,a_{i+1})
		=
		i+i\bigl(|a_1|+|a_{i+1}|\bigr)
		+|a_{i+1}|\bigl(|a_1|+\cdots+|a_i|\bigr).
	\end{equation*}
\end{proposition}
\begin{proof}
	By linearity, it suffices to consider homogeneous elements
	$a_j=\omega_j\otimes l_j$, where $\omega_j\in\Omega^\bullet(M)$ and $l_j\in\mathfrak L$.
	All sign exponents below are understood modulo 2.
	We first consider the case $i=1$. Since $M$ has no boundary, Stokes'
	theorem gives
	\begin{equation*}
		\int_M \omega_1 \wedge d\omega_2=(-1)^{|\omega_1| +1}\int_M d\omega_1 \wedge \omega_2.
	\end{equation*}
	Applying \eqref{mu1}, \eqref{pair}, graded symmetry of $\langle-,-\rangle$, and the $i=1$ instance of \eqref{cyc}, we obtain
	\begin{align*}
		\ll a_1, \mu^*_1(a_2)\gg
		=&	\ll \omega_1 \otimes l_1, d \omega_2\otimes l_2 +(-1)^{|\omega_2|}\omega_2 \otimes \mu_1(l_2)\gg\\
		=&(-1)^{|d\omega_2| |l_1|}\int_M \omega_1 \wedge d\omega_2 \langle l_1, l_2\rangle + (-1)^{|\omega_2| +|\omega_2| |l_1|}\int_M \omega_1 \wedge \omega_2 \langle l_1, \mu_1(l_2)\rangle\\
		=&(-1)^{(|\omega_2|+1)|l_1| +|l_1| |l_2| +1 +|\omega_1| +|\omega_2|(|\omega_1|+1)}\int_M \omega_2\wedge d\omega_1 \langle l_2, l_1\rangle \\
		&+(-1)^{|\omega_2|+ |\omega_2| |l_1| +1 +|l_1| +|l_2| +|l_1| |l_2| +|\omega_1| |\omega_2|}\int_M \omega_2\wedge \omega_1 \langle l_2, \mu_1(l_1)\rangle.
	\end{align*}
	On the other hand, an analogous calculation gives
	\begin{align*}
		\ll a_2, \mu^*_1(a_1)\gg 
		=& \ll \omega_2\otimes l_2, d\omega_1 \otimes l_1+(-1)^{|\omega_1|}\omega_1 \otimes \mu_1(l_1)\gg\\
		=&(-1)^{|d\omega_1| |l_2|}\int_M \omega_2 \wedge d\omega_1 \langle l_2, l_1\rangle + (-1)^{|\omega_1| +|\omega_1| |l_2|}\int_M \omega_2 \wedge \omega_1 \langle l_2, \mu_1(l_1)\rangle.
	\end{align*}
	The difference between the corresponding exponents is
	\begin{align*}
		\xi_1(a_1,a_2)
		&=
		1+|\omega_1| +|l_1| +|\omega_2| +|l_2|+(|\omega_1| +|l_1| )(|\omega_2| +|l_2| )
		\\
		&=
		1+|a_1|+|a_2|+|a_1||a_2|.
	\end{align*}
	Therefore,
	\begin{align*}
		\ll a_1, \mu^*_1(a_2)\gg=(-1)^{\xi_1(a_1,a_{2})}	\ll a_2, \mu^*_1(a_1)\gg.
	\end{align*}
	
	Now let $i\geq 2$. By \eqref{cyc},  \eqref{mun}, and \eqref{pair}, we then have
	\begin{align*}
		&\ll a_1,\mu_i^*(a_2,\dots,a_{i+1})\gg\\
		=&
		(-1)^{i\sum\limits_{j=2}^{i+1}|\omega_j| + \sum\limits_{j=-1}^{i-3}|\omega_{i-j}|\sum\limits_{k=2}^{i-j-1}|l_k|+|l_1|\sum\limits_{j=2}^{i+1}|\omega_j|}
		\int_M \omega_1\wedge\cdots\wedge\omega_{i+1}\;
		\big\langle l_1,\mu_i(l_2,\dots,l_{i+1})\big\rangle\\
		=&
		(-1)^{i\sum\limits_{j=2}^{i+1}|\omega_j| + \sum\limits_{j=-1}^{i-3}|\omega_{i-j}|\sum\limits_{k=2}^{i-j-1}|l_k|+|l_1|\sum\limits_{j=2}^{i+1}|\omega_j|+i+i(|l_1|+|l_{i+1}|)+|l_{i+1}|\sum\limits_{j=1}^{i}|l_j|}\\
		&
		\int_M \omega_1\wedge\cdots\wedge\omega_{i+1}\;
		\big\langle l_{i+1},\mu_i(l_1,\dots,l_i)\big\rangle\\
		=&
		(-1)^{i\sum\limits_{j=2}^{i+1}|\omega_j| + \sum\limits_{j=-1}^{i-3}|\omega_{i-j}|\sum\limits_{k=2}^{i-j-1}|l_k|+|l_1|\sum\limits_{j=2}^{i+1}|\omega_j|+i+i(|l_1|+|l_{i+1}|)+|l_{i+1}|\sum\limits_{j=1}^{i}|l_j|+|\omega_{i+1}|\sum\limits_{j=1}^{i}|\omega_j|}\\
		&
		\int_M \omega_{i+1}\wedge\omega_1\wedge\cdots\wedge\omega_{i}\;
		\big\langle l_{i+1},\mu_i(l_1,\dots,l_i)\big\rangle,
	\end{align*}
	and 
	\begin{align*}
		\ll a_{i+1},\mu_i^*(a_1,\dots,a_i)\gg=&\ll \omega_{i+1}\otimes l_{i+1},\mu_i^*(\omega_1\otimes l_1,\dots,\omega_{i}\otimes l_{i})\gg \\
		=&(-1)^{i\sum\limits_{j=1}^{i}|\omega_j| +\sum\limits_{j=0}^{i-2}|\omega_{i-j}|\sum\limits_{k=1}^{i-j-1}|l_k| +|l_{i+1}|\sum\limits_{j=1}^{i}|\omega_j|}	\\
		&\int_M \omega_{i+1}\wedge\omega_1\wedge\cdots\wedge\omega_{i}\;
		\big\langle l_{i+1},\mu_i(l_1,\dots,l_i)\big\rangle.
	\end{align*}
	A straightforward rearrangement of the exponents gives
	\begin{align*}
		&i\sum_{j=2}^{i+1}|\omega_j| + \sum_{j=-1}^{i-3}|\omega_{i-j}|\sum_{k=2}^{i-j-1}|l_k|+|l_1|\sum_{j=2}^{i+1}|\omega_j|+i+i(|l_1|+|l_{i+1}|)+|l_{i+1}|\sum_{j=1}^{i}|l_j|+|\omega_{i+1}|\sum_{j=1}^{i}|\omega_j|\\
		=&\xi_i(a_1,\dots,a_{i+1}) +i\sum_{j=1}^{i}|\omega_j| +\sum_{j=0}^{i-2}|\omega_{i-j}|\sum_{k=1}^{i-j-1}|l_k| +|l_{i+1}|\sum_{j=1}^{i}|\omega_j|.
	\end{align*}
	Substituting this identity into the preceding expressions proves
	\eqref{eq:cyc-induced}.
\end{proof}

\subsection{Homotopy Maurer--Cartan theory}\label{se-HMC}
Homotopy MC theory provides an $L_\infty$-algebraic extension
of CS theory,  with HCS theories arising as distinguished examples. Following \cite{BLCM}, we recall below the relevant definitions and conventions. For further background on
homotopy MC theory and higher gauge theory, see
\cite{GR,BTLCM}.

Let $(\mathfrak{L},\{\mu_i\}_{i\ge 1})$ be an $L_\infty$-algebra. For
every smooth manifold $M$, the de Rham complex $\Omega^\bullet(M,\mathfrak{L})$ carries the induced $L_\infty$-structure $(\Omega^\bullet(M,\mathfrak L),
\{\mu_i^*\}_{i\geq1})$. A homogeneous element
$\mathcal A\in \Omega^\bullet_1(M,\mathfrak L)$
is called a \emph{gauge potential}. Its \emph{curvature} is the element $\mathcal F\in\Omega^\bullet_2(M,\mathfrak L)$ defined by
\begin{equation}\label{cur}
	\mathcal F
	\coloneqq \sum_{i\ge1}\frac{1}{i!}\,\mu_i^*(\mathcal A,\dots,\mathcal A)
	=\mu_1^*(\mathcal A)+\frac12\,\mu_2^*(\mathcal A,\mathcal A)+\cdots.
\end{equation}
The gauge potential $\mathcal{A}$ is called an \emph{MC element} if its curvature vanishes. Equivalently, it satisfies the \emph{homotopy MC equation}
\begin{equation*}
	\sum_{i\ge1}\frac{1}{i!}\,\mu_i^*(\mathcal A,\dots,\mathcal A)=0.
\end{equation*}
The higher homotopy Jacobi identities \eqref{HJI} imply the Bianchi identity
\begin{equation}\label{BI}
	\sum_{i\ge0}\frac{(-1)^i}{i!}\,
	\mu_{i+1}^*\bigl(\mathcal F,\underbrace{\mathcal A,\dots,\mathcal A}_{i\ \mathrm{copies}}\bigr)=0.
\end{equation}
For an ordinary Lie algebra, these definitions reduce to the usual
notions of connection, curvature, and the Bianchi identity. More
generally, for Lie $n$-algebras they reproduce their higher
gauge-theoretic analogues.

An infinitesimal gauge transformation is parametrized by a homogeneous
element $\lambda_0\in \Omega^\bullet_0(M,\mathfrak L)$, and its action on 
$\mathcal{A}$ is defined by
\begin{equation}\label{gtraA}
	\delta_{\lambda_0}\mathcal A
	\coloneqq \sum_{i\ge0}\frac{1}{i!}\,
	\mu^*_{i+1}\bigl(\underbrace{\mathcal A,\dots,\mathcal A}_{i\ \mathrm{copies}},\lambda_0\bigr).
\end{equation}
The corresponding variation of the curvature is given by
\begin{equation}\label{gtraf}
	\delta_{\lambda_0}\mathcal F
	=\sum_{i\ge0}\frac{1}{i!}\,
	\mu^*_{i+2}\bigl(\underbrace{\mathcal A,\dots,\mathcal A}_{i\ \mathrm{copies}},\mathcal F,\lambda_0\bigr).
\end{equation}
Applying the higher homotopy Jacobi identities, one obtains
\begin{equation*}
	[\delta_{\lambda_0}, \delta_{\lambda'_0}]\mathcal{A}=\delta_{\lambda''_0}\mathcal{A}+\sum_{i\geq0}\dfrac{1}{i!}\mu^*_{i+3}(\mathcal{A}, \cdots, \mathcal{A}, \mathcal F, \lambda_0, \lambda'_0),
\end{equation*}
where the composite parameter is 
\begin{equation*}
	\lambda''_0\coloneqq \sum_{i\geq 0}\dfrac{1}{i!}\mu^*_{i+2}(\mathcal{A}, \cdots, \mathcal{A}, \lambda_0, \lambda'_0).
\end{equation*}
Thus, the infinitesimal gauge algebra is generally open: it closes on
MC elements, i.e.,\ on configurations satisfying $\mathcal F=0$. For 1-term $L_\infty$-algebras, where $\mu_i=0$ for all $i\ge 3$, the closure is exact.
Besides, one may introduce higher gauge parameters $\lambda_{-k}\in\Omega^\bullet_{-k}(M,\mathfrak L)$ and define the corresponding gauge-for-gauge transformations recursively by
\begin{equation*}
	\delta_{\lambda_{-k-1}}\lambda_{-k}
	\coloneqq \sum_{i\ge0}\frac{1}{i!}\,
	\mu^*_{i+1}\bigl(\underbrace{\mathcal A,\dots,\mathcal A}_{i\ \mathrm{copies}},\lambda_{-k-1}\bigr),
	\qquad k\ge0.
\end{equation*}
These higher transformations are also closed when $\mathcal F=0$. In the present work,
we only consider the lowest-level transformation \eqref{gtraA}. 
A systematic discussion of the higher gauge transformations can be in found in
\cite{BLCM}.

We next introduce the homotopy MC action.
Let $M$ be a compact oriented $n$-dimensional manifold without boundary with $n\geq 3$, and let $(\mathfrak L,\{\mu_i\}_{i\ge 1},\langle-,-\rangle)$ be a cyclic $(n-2)$-term $L_\infty$-algebra equipped with an invariant pairing of degree $n-3$. By \eqref{pair} and
Proposition~\ref{lem:cyclic-induced}, the induced pairing $\ll-,-\gg$ endows $\Omega^\bullet(M,\mathfrak L)$ with a cyclic $L_\infty$-structure
$$
\bigl(\Omega^\bullet(M,\mathfrak L),\{\mu_i^*\}_{i\ge 1},\ll-,-\gg\bigr),
$$
whose pairing has degree $-3$.
For a gauge potential $\mathcal A\in\Omega^\bullet_1(M,\mathfrak L)$,  define the homotopy
MC action by
\begin{equation}\label{HCS}
	S_n(\mathcal A)
	\coloneqq \sum_{i\ge 1}\frac{1}{(i+1)!}\,
	\ll \mathcal A,\mu_i^*(\mathcal A,\dots,\mathcal A)\gg.
\end{equation}
We refer to \eqref{HCS} as the HCS action associated
with $\mathfrak L$; see \cite{GR}.
The cyclicity of $\ll-,-\gg$ implies that
\begin{equation*}
	\delta S_n(\mathcal A)
	=
	\ll\delta\mathcal A,\mathcal F\gg.
\end{equation*}
Consequently, the equation of motion of \eqref{HCS}  is precisely the homotopy
MC equation $\mathcal F=0$.
Moreover, $S_n(\mathcal A)$ is invariant off shell under the infinitesimal gauge transformation \eqref{gtraA}. Indeed, 
\begin{align*}
	\delta_{\lambda_0}S_n(\mathcal A)
	=-\sum_{i\ge 0}\frac{(-1)^i}{i!}\,
	\ll \lambda_0,\mu^*_{i+1}(\mathcal F,\mathcal A,\dots,\mathcal A)\gg
	=\ll \mathcal F,\delta_{\lambda_0}\mathcal A\gg
	=0,
\end{align*}
where the second equality follows from the graded symmetry and cyclicity
of the induced pairing, and the last equality is a consequence of the Bianchi identity \eqref{BI}. Notice that this off-shell invariance of the
action is compatible with the fact that the corresponding gauge algebra
is generally open and closes only on solutions of the field equation.

The following example shows that the three-dimensional CS functional is recovered as a special case of \eqref{HCS}. Thus, ordinary CS theory arises from the homotopy
MC construction when the underlying cyclic $L_\infty$-algebra is an ordinary Lie algebra. Higher examples, including the 2-term case, will be discussed in
Section~\ref{sec:2term}.
\begin{example}\label{ex:3d-CS}
	Let $n=3$, and let $\mathfrak L$ be concentrated in degree zero, i.e., $	\mathfrak L=\mathfrak L_0$.
	Degree considerations imply that $\mu_1=0$ and $\mu_i=0$ for $i\geq 3$. Hence $\mathfrak L$ is an ordinary Lie algebra with Lie bracket  $\mu_2=[-,-]$.
	A cyclic structure on $\mathfrak L$ is specified by a non-degenerate
	symmetric bilinear form $\langle-,-\rangle$ on $\mathfrak L_0$ satisfying
	\begin{equation*}
		\langle [X_1,X_2],X_3\rangle=-\langle X_2,[X_1,X_3]\rangle, \qquad \forall  X_1,X_2,X_3\in\mathfrak L_0.
	\end{equation*}
	A gauge potential is an $\mathfrak L_0$-valued one-form  $A\in\Omega^1(M)\otimes\mathfrak L_0$, and its  curvature \eqref{cur} reduces to
	\begin{equation*}
		F=dA+\frac12[A,A]\in\Omega^2(M)\otimes\mathfrak L_0,
	\end{equation*}
	where $[-,-]$ is the induced bracket on $\Omega^{\bullet}(M)\otimes\mathfrak L_0$ given by
	\begin{equation*}
		[A_1\otimes X_1,A_2\otimes X_2]=(A_1\wedge A_2)\otimes [X_1,X_2]
	\end{equation*}
	for homogeneous forms $A_1, A_2$.
	In this case, \eqref{gtraA} becomes the usual infinitesimal gauge
	transformation
	\begin{equation}\label{eq:3d-gauge-tra-A}
		\delta_{\lambda_0}A=d\lambda_0+[A,\lambda_0],
		\qquad
		\lambda_0\in\Omega^0(M)\otimes\mathfrak L_0,
	\end{equation}
	and the curvature transforms as
	\begin{equation*}
		\delta_{\lambda_0}F=[F,\lambda_0].
	\end{equation*}
	Accodingly, the HCS action \eqref{HCS} reduces to the action of standard CS theory
	\begin{equation}\label{eq:3d-CS-action}
		S_{3}(A)
		=
		\int_M\Bigl(
		\tfrac12\langle A,dA\rangle
		+\tfrac1{3!}\langle A,[A,A]\rangle
		\Bigr),
	\end{equation}
	where the cyclic pairing is extended to $\Omega^\bullet(M)\otimes\mathfrak L_0$ by
	\begin{equation*}
		\langle A_1\otimes X_1, A_2\otimes X_2\rangle= A_1\wedge A_2\langle X_1, X_2\rangle.
	\end{equation*}
	Being independent of any metric on $M$, the action is topological. The equation of motion is readily identified with the flatness condition $F=0$. In addition, the functional \eqref{eq:3d-CS-action} is invariant under the infinitesimal gauge transformations \eqref{eq:3d-gauge-tra-A}.
\end{example}

\section{2-Term $L_\infty$-Algebra Gauge Theory}\label{sec:2term}
In this section, we recall gauge theories associated with 2-term $L_\infty$-algebras, which are the main objects of this work.  Starting from the homotopy MC formalism introduced in the previous section, we describe the
corresponding higher gauge fields, curvature forms, and infinitesimal
gauge transformations. We then specialize the construction to the
four-dimensional HCS theory induced by a 2-term $L_\infty$-algebra. For a detailed account of semistrict higher gauge
theory, see \cite{BLCM}.

\subsection{Balanced 2-term $L_\infty$-algebras and invariant pairings}
A 2-term $L_\infty$-algebra is an $L_\infty$-algebra, whose
underlying graded vector space is concentrated in degrees $-1$ and $0$:
\begin{equation*}
	\mathfrak L=\mathfrak L_{-1}\oplus \mathfrak L_0.
\end{equation*}
Since $\text{deg}(\mu_i)=2-i$, the only  non-vanishing structure
maps are, up to graded antisymmetry,
\begin{equation*}
	\begin{aligned}
		\mu_1&\colon \mathfrak L_{-1}\to \mathfrak L_0,\\
		\mu_2&\colon \mathfrak L_0\times \mathfrak L_0\to \mathfrak L_0,\\
		\mu_2&\colon \mathfrak L_0\times \mathfrak L_{-1}\to \mathfrak L_{-1},\\
		\mu_3&\colon \mathfrak L_0\times \mathfrak L_0 \times \mathfrak L_0\to \mathfrak L_{-1}.
	\end{aligned}
\end{equation*}
In particular, $\mu_i=0$ for all $i\ge 4$. For $X,X_i\in\mathfrak L_{0}$ and $Y\in\mathfrak L_{-1}$, graded antisymmetry implies
\begin{subequations}\label{q51}
	\begin{align}
		\mu_2(X_1,X_2)&=-\mu_2(X_2,X_1),\label{eq:antisym-00}\\
		\mu_2(X,Y)&=-\mu_2(Y,X),\label{eq:antisym-0-1}\\
		\mu_3(X_1,X_2,X_3)&=-\mu_3(X_2,X_1,X_3)=\mu_3(X_2,X_3,X_1).\label{eq:antisym-3}
	\end{align}
\end{subequations}
The higher Jacobi identities \eqref{HJI} reduce to a finite collection
of non-trivial relations. For $X,X_i\in\mathfrak L_{0}$ ($i=1,\dots,4$) and $Y,Y_j\in\mathfrak L_{-1}$ ($j=1,2$), these relations are
\begin{subequations}\label{q61}
	\begin{align}
		&\mu_2(X, \mu_1(Y))-\mu_1(\mu_2(X, Y))=0,\label{eq1}\\
		&\mu_2(\mu_1(Y_1), Y_2)+\mu_2(\mu_1(Y_2), Y_1)=0,\label{eq2}\\
		&\mu_2(X_1, \mu_2(X_2, X_3))+\mu_2(X_2, \mu_2(X_3, X_1))+\mu_2(X_3, \mu_2(X_1, X_2))+\mu_1(\mu_3(X_1, X_2, X_3))=0,\label{eq3}\\
		&\mu_2(X_1, \mu_2(X_2, Y))-\mu_2(X_2, \mu_2(X_1, Y))-\mu_2(\mu_2(X_1, X_2), Y)+\mu_3(X_1, X_2, \mu_1(Y))=0,\label{eq4}\\
		&\mu_3(X_1, X_2, \mu_2(X_3, X_4))+\mu_3(X_1, X_3, \mu_2(X_4, X_2))+\mu_3(X_1, X_4, \mu_2(X_2, X_3))
		-\mu_3(X_2, X_3, \mu_2(X_4, X_1))\nonumber\\
		&-\mu_3(X_3, X_4, \mu_2(X_2, X_1))-\mu_3(X_4, X_2, \mu_2(X_3, X_1))
		-\mu_2(X_1, \mu_3(X_2, X_3, X_4))\nonumber\\
		&+\mu_2(X_2, \mu_3(X_3, X_4, X_1))	-\mu_2(X_3, \mu_3(X_4, X_1, X_2))
		+\mu_2(X_4, \mu_3(X_1, X_2, X_3))=0.	\label{eq5}
	\end{align}
\end{subequations}

In \cite{JCB-ASC}, it is shown that the $2$-term $L_\infty$-algebras are categorically equivalent to semistrict Lie $2$-algebras. A $2$-term $L_\infty$-algebra is called \emph{strict} if its trilinear bracket vanishes, namely, if $\mu_3=0$. In this case, it is equivalent to a crossed module of Lie algebras. The correspondence is illustrated by the following example.
\begin{example}\label{strict}
	Let $\mathcal{h}$ and $\mathcal{g}$ be Lie algebras, 
	$\alpha:\mathcal{h}\to\mathcal{g}$
	be a Lie algebra homomorphism, and  $\vartriangleright:\mathcal{g}\times\mathcal{h}\to\mathcal{h}$
	be an action of $\mathcal{g}$ on $\mathcal{h}$ by derivations. The quadruple $(\mathcal{h},\mathcal{g},\alpha,\vartriangleright)$ is a \emph{crossed module of Lie algebras} if 
	\begin{equation*}
		\alpha(X\vartriangleright Y)=[X,\alpha(Y)],
		\qquad
		\alpha(Y_1)\vartriangleright Y_2=[Y_1,Y_2],\qquad \forall X\in\mathcal{g}, Y,Y_1,Y_2\in\mathcal{h}.
	\end{equation*}
	These identities are referred to as the equivariance condition and the Peiffer identity, respectively.
	Set $\mathfrak L_{-1}\coloneqq \mathcal{h}$ and $ \mathfrak L_0\coloneqq \mathcal{g}$,
	and then the corresponding strict 2-term $L_\infty$-structure is given by
	\begin{align*}
		\mu_1(Y)&\coloneqq \alpha(Y), \qquad \mu_1(X)\coloneqq 0,\\
		\mu_2(X_1,X_2)&\coloneqq [X_1,X_2],\qquad \mu_2(Y_1,Y_2)\coloneqq 0,\\
		\mu_2(X,Y)&\coloneqq X\vartriangleright Y,\qquad  \mu_2(Y,X)\coloneqq -\,\mu_2(X,Y),
	\end{align*}
	for $X,X_1,X_2\in\mathfrak{L}_0$ and $Y,Y_1,Y_2\in\mathfrak{L}_{-1}$. 
	The crossed module identities are precisely the non-trivial  $L_\infty$-identities in the strict case. Conversely, every strict $2$-term $L_\infty$-algebra arises from a crossed module of Lie algebras in this way.
\end{example}

We now specialize the homotopy MC action \eqref{HCS} to $n=4$.
In this case, the invariant pairing entering the action has
degree $1$. Throughout this section, $\mathfrak L$  is assumed
to be finite-dimensional.
\begin{proposition}\label{lem:2term-degree1-pairing}
	Let $\langle-,-\rangle$ be a graded symmetric bilinear pairing of degree $1$ on a $2$-term $L_\infty$-algebra $\mathfrak L$. Then, 
	\begin{equation*}
		\langle X, X'\rangle=\langle Y, Y' \rangle=0, \qquad \forall X, X' \in \mathfrak{L}_0, Y, Y'\in \mathfrak{L}_{-1}.
	\end{equation*}
	Hence  $\langle-,-\rangle$  is  a mixed pairing between $\mathfrak L_0$ and $\mathfrak L_{-1}$, and it induces a non-degenerate pairing
	\begin{equation*}
		\langle-,-\rangle:\mathfrak L_0\times \mathfrak L_{-1}\to\mathbb R,
	\end{equation*}
	and
	\begin{equation*}
		\mathfrak{L}_0\cong\mathfrak{L}_{-1}^*,\qquad \dim\mathfrak{L}_0=\dim\mathfrak{L}_{-1},
	\end{equation*}
	if it is non-degenerate.
\end{proposition}
\begin{proof}
	Since $\langle-,-\rangle$ has degree $1$ and $\mathbb{R}$ is concentrated in degree $0$, the pairing $\langle X,Y\rangle$ can be non-zero only if $$|X|+|Y|=-1.$$
	As $\mathfrak L$ is concentrated in degrees $-1$ and 0, this
	excludes the components  $\langle\mathfrak{L}_0,\mathfrak{L}_0\rangle$ and $\langle\mathfrak{L}_{-1},\mathfrak{L}_{-1}\rangle$. Thus only the
	mixed pairing between $\mathfrak L_0$ and $\mathfrak L_{-1}$ can be
	non-vanishing. Its non-degeneracy yields $\mathfrak{L}_0\cong\mathfrak{L}_{-1}^*$, and hence $\dim\mathfrak{L}_0=\dim\mathfrak{L}_{-1}$.
\end{proof}

\begin{definition}\label{def:balanced-2term}
	A $2$-term $L_\infty$-algebra $\mathfrak L=\mathfrak L_{-1}\oplus\mathfrak L_0$ is called \emph{balanced} if 
	$\dim\mathfrak L_0=\dim\mathfrak L_{-1}$.
\end{definition}

\begin{definition}\label{def:2term-invariant-pairing}
	Let $\mathfrak L=\mathfrak L_{-1}\oplus\mathfrak L_0$ be a balanced $2$-term $L_\infty$-algebra.
	An \emph{invariant pairing} on $\mathfrak L$ is a non-degenerate bilinear map
	\begin{equation*}
		\langle-,-\rangle:\mathfrak L_0\times\mathfrak L_{-1}\to\mathbb R
	\end{equation*}
	such that, for all $X,X',X_1,X_2\in\mathfrak L_0$ and $Y,Y_1,Y_2\in\mathfrak L_{-1}$, 
	\begin{subequations}\label{q14}
		\begin{align}
			\langle \mu_1(Y_1),Y_2\rangle-\langle \mu_1(Y_2),Y_1\rangle&=0, \label{eq:invpair-mu1}\\
			\langle \mu_2(X_1,X_2),Y\rangle+\langle X_2,\mu_2(X_1,Y)\rangle&=0, \label{eq:invpair-mu2}\\
			\langle X_1,\mu_3(X,X',X_2)\rangle+\langle X_2,\mu_3(X,X',X_1)\rangle&=0. \label{eq:invpair-mu3}
		\end{align}
	\end{subequations}
\end{definition}
The identities \eqref{eq:invpair-mu1}--\eqref{eq:invpair-mu3} are the
component form of the cyclicity condition \eqref{cyc} for a $2$-term $L_\infty$-algebra. By Proposition~\ref{lem:2term-degree1-pairing},
balancedness is necessary for the existence of a non-degenerate
invariant pairing of degree 1,  although it is not sufficient.  A non-balanced $2$-term $L_\infty$-algebra may be embedded into a
minimal balanced extension; see \cite{OZ}. In the remainder of this
paper, we restrict attention to balanced $2$-term $L_\infty$-algebras 
equipped with a non-degenerate invariant pairing
of degree 1.

\subsection{Semistrict higher gauge theory}\label{subsec:higher-gauge-theory}
Throughout this subsection, let $M$ be a compact oriented manifold without boundary, and let $(\mathfrak L,\{\mu_i\}_{i\ge 1},\langle-,-\rangle)$ be a cyclic $2$-term $L_\infty$-algebra. Equivalently, $	\mathfrak L=\mathfrak L_{-1}\oplus \mathfrak L_0$  is a balanced 2-term $L_\infty$-algebra equipped with a non-degenerate
invariant pairing of degree 1 in the sense of
Definition~\ref{def:2term-invariant-pairing}.
We extend the de Rham differential to $\Omega^\bullet(M,\mathfrak L)$ by 
\begin{equation*}
	d(\omega\otimes  l )\coloneqq (d\omega)\otimes  l 
\end{equation*}
for homogeneous $\omega\in\Omega^\bullet(M)$ and $ l \in\mathfrak L$,  and then by linearity. The pairing on $\mathfrak L$ induces a natural pairing on $\Omega^\bullet(M,\mathfrak L)$,
\begin{equation*}
	\ll \omega_1\otimes  l _1,\omega_2\otimes  l _2\gg
	=\int_M \omega_1\wedge \omega_2\,\langle  l _1, l _2\rangle
	=\int_M \big\langle \omega_1\otimes  l _1,\omega_2\otimes  l _2\big\rangle,
\end{equation*}
where the associated integrated pairing is defined by
\begin{equation}\label{omel}
	\big\langle \omega_1\otimes  l _1,\omega_2\otimes  l _2\big\rangle
	\coloneqq \omega_1\wedge \omega_2\,\langle  l _1, l _2\rangle.
\end{equation}
Similarly, the higher brackets extend to $\Omega^\bullet(M,\mathfrak L)$ by
\begin{equation}\label{eq:mu-extension-convention-forms}
	\mu_i\bigl(\omega_1\otimes  l _1,\dots,\omega_i\otimes  l _i\bigr)
	\coloneqq 
	(\omega_1\wedge\cdots\wedge \omega_i)\otimes \mu_i( l _1,\dots, l _i)
\end{equation}
for homogeneous elements, and then by multilinearity. In the
following, we use the same notations $\langle-,-\rangle$ and $\mu_i$ for these extended operations.

With the tensor-product sign convention adopted in
\eqref{mu1} and \eqref{mun}, the induced operations $\mu_i^*$ on $\Omega^\bullet(M,\mathfrak L)$ take the forms
\begin{align}
	\mu^*_1(\omega_1\otimes l_1)&=d(\omega_1 \otimes l_1) +(-1)^{|\omega_1| } \mu_1(\omega_1\otimes l_1),\label{mustar1}\\
	\mu^*_i(\omega_1\otimes l_1, \dots, \omega_i\otimes l_i)&=(-1)^{i\sum\limits_{j=1}^{i}|\omega_j|  +\sum\limits_{j=0}^{i-2}|\omega_{i-j}| \sum\limits_{k=1}^{i-j-1}|l_k|}	\mu_i\bigl(\omega_1\otimes l_1,\dots,\omega_i\otimes l_i\bigr).\label{mustar2}
\end{align}


\begin{proposition}\label{lem:graded-antisym-forms}
	Let $\mathfrak L=\mathfrak L_{-1}\oplus\mathfrak L_0$ be a $2$-term $L_\infty$-algebra, and let
	$\omega,\omega_i,\eta\in\Omega^\bullet(M)$ be homogeneous.
	For $X,X_i\in\mathfrak L_0$ and $Y\in\mathfrak L_{-1}$, the extended brackets satisfy 
	\begin{align*}
		\mu_2(\omega_1\otimes X_1,\omega_2\otimes X_2)
		&=-(-1)^{|\omega_1||\omega_2|}\,
		\mu_2(\omega_2\otimes X_2,\omega_1\otimes X_1),\\
		\mu_2(\omega\otimes X,\eta\otimes Y)
		&=-(-1)^{|\omega||\eta|}\,
		\mu_2(\eta\otimes Y,\omega\otimes X),\\
		\mu_3(\omega_1\otimes X_1,\omega_2\otimes X_2,\omega_3\otimes X_3)
		&=-(-1)^{|\omega_1||\omega_2|}\,
		\mu_3(\omega_2\otimes X_2,\omega_1\otimes X_1,\omega_3\otimes X_3)\nonumber\\
		&=(-1)^{|\omega_1|(|\omega_2|+|\omega_3|)}\,
		\mu_3(\omega_2\otimes X_2,\omega_3\otimes X_3,\omega_1\otimes X_1).
	\end{align*}
\end{proposition}
\begin{proof}
	The results follow directly from
	\eqref{eq:mu-extension-convention-forms}, the graded antisymmetry
	relations \eqref{q51}, and the graded
	commutativity of the wedge product.
\end{proof}

\begin{proposition}\label{lem:cyclic-pairing-forms}
	Let $\mathfrak L=\mathfrak L_{-1}\oplus\mathfrak L_0$ be a cyclic $2$-term $L_\infty$-algebra.
	For homogeneous $\omega,\omega',\omega_i,\eta\in\Omega^\bullet(M)$, $X,X',X_i\in\mathfrak L_0$, and
	$Y,Y_i\in\mathfrak L_{-1}$ ($i=1, 2$),  the extended pairing satisfies
	\begin{subequations}\label{eq:invpair-forms}
		\begin{align}
			\big\langle \mu_1(\omega_1\otimes Y_1),\,\omega_2\otimes Y_2\big\rangle
			-(-1)^{|\omega_1||\omega_2|}\,
			\big\langle \mu_1(\omega_2\otimes Y_2),\,\omega_1\otimes Y_1\big\rangle
			&=0,\label{eq:invpair-mu1-forms}\\
			\big\langle \mu_2(\omega_1\otimes X_1,\,\omega_2\otimes X_2),\,\eta\otimes Y\big\rangle
			+(-1)^{|\omega_1||\omega_2|}\,
			\big\langle \omega_2\otimes X_2,\,\mu_2(\omega_1\otimes X_1,\,\eta\otimes Y)\big\rangle
			&=0,\label{eq:invpair-mu2-forms}\\
			\big\langle \omega_1\otimes X_1,\,\mu_3(\omega\otimes X,\,\omega'\otimes X',\,\omega_2\otimes X_2)\big\rangle
			+(-1)^{|\omega_1||\omega_2|+(|\omega_1|+|\omega_2|)(|\omega|+|\omega'|)}&\, \nonumber\\
			\big\langle \omega_2\otimes X_2,\,\mu_3(\omega\otimes X,\,\omega'\otimes X',\,\omega_1\otimes X_1)\big\rangle
			&=0.\label{eq:invpair-mu3-forms}
		\end{align}
	\end{subequations}
\end{proposition}
\begin{proof}
	The identities are direct consequences of
	\eqref{omel}, \eqref{eq:mu-extension-convention-forms}, the invariance
	relations \eqref{q14}, and the
	graded commutativity of the wedge product.
\end{proof}

\subsubsection{Higher gauge fields and infinitesimal gauge transformations}
We now specialize the homotopy MC formalism of higher gauge theory in
Section~\ref{se-HMC} to the $2$-term $L_\infty$-algebra $\mathfrak{L}$. The
total-degree-one component of $\Omega^{\bullet}(M, \mathfrak{L})$ is
$\Omega^{\bullet}_1(M, \mathfrak{L})=(\Omega^1(M)\otimes \mathfrak L_0)\oplus (\Omega^2(M)\otimes \mathfrak L_{-1})$.
Accordingly, 
a gauge potential on $M$ with values in $\mathfrak{L}$ is  a pair
\begin{equation*}
	\mathcal A=A+B\in\Omega^\bullet_1(M,\mathfrak L),
	\qquad
	A\in\Omega^1(M)\otimes \mathfrak L_0,\qquad
	B\in\Omega^2(M)\otimes \mathfrak L_{-1}.
\end{equation*}
We refer to $(A, B)$ as a \emph{$2$-connection}.
Its curvature is the total-degree-two element
\begin{equation}\label{cur-d}
	\mathcal F
	=\mu^*_1(\mathcal A)+\frac12\,\mu^*_2(\mathcal A,\mathcal A)+\frac16\,\mu^*_3(\mathcal A,\mathcal A,\mathcal A)
	\in\Omega^\bullet_2(M,\mathfrak L).
\end{equation}
Using \eqref{mustar1} and \eqref{mustar2}, we obtain
\begin{align*}
	\mu_1^*(\mathcal A)&=dA+dB+\mu_1(B),\\
	\mu_2^*(\mathcal A,\mathcal A)&=\mu_2(A,A)+\mu_2(A,B)-\mu_2(B,A),\\
	\mu_3^*(\mathcal A,\mathcal A,\mathcal A)&=-\,\mu_3(A,A,A).\
\end{align*}
Here $\mu_2(B,B)=0$, since every term in $\mu_3^*(\mathcal A,\mathcal A,\mathcal A)$ containing at least one
factor of $B$ vanishes for degree reasons. Hence, the curvature \eqref{cur-d} is reduced to
\begin{equation*}
	\mathcal F=F+H,\qquad
	F\in\Omega^2(M)\otimes\mathfrak L_0,\qquad
	H\in\Omega^3(M)\otimes\mathfrak L_{-1},
\end{equation*}
where
\begin{align}\label{eq:2term-curvature-components-FH}
	F= dA+\frac12\,\mu_2(A,A)+\mu_1(B),\qquad
	H= dB+\mu_2(A,B)-\frac16\,\mu_3(A,A,A).
\end{align}
The pair $(F,H)$ is called the \emph{$2$-curvature} of  $(A,B)$, with $F$ being the fake curvature of the 2-connection. The 2-connection  is said to be fake-flat if
if $F=0$, and flat if in addition $H=0$.

In this case, the general Bianchi identity \eqref{BI} accordingly reduces to
\begin{equation*}
	\mu^*_1(\mathcal F)-\mu^*_2(\mathcal F,\mathcal A)+\frac12\,\mu^*_3(\mathcal F,\mathcal A,\mathcal A)=0.
\end{equation*}
The two homogeneous components of this identity can be obtained directly from the induced operations $,\mu^*_i$ using \eqref{mustar1} and \eqref{mustar2}. In what follows, we verify them directly from the defining expressions~\eqref{eq:2term-curvature-components-FH}.
\begin{theorem}\label{thm:2-bianchi}
	Let $(A,B)$ be a $2$-connection with values in a $2$-term $L_\infty$-algebra, and let $(F,H)$ be its $2$-curvature defined by~\eqref{eq:2term-curvature-components-FH}. Then
	\begin{subequations}\label{eq:2-bianchi}
		\begin{align}
			dF+\mu_2(A,F)-\mu_1(H)&=0,\label{BI1}\\
			dH+\mu_2(A,H)-\mu_2(F,B)+\frac12\,\mu_3(A,A,F)&=0,\label{BI2}
		\end{align}
	\end{subequations}
	which are called the \emph{$2$-Bianchi identities}.
\end{theorem}
\begin{proof}
	We start with the structure equations
	\begin{align*}
		dA=F-\frac12\,\mu_2(A,A)-\mu_1(B),\qquad dB=H-\mu_2(A,B)+\frac16\,\mu_3(A,A,A).
	\end{align*}
	Applying the exterior derivative $d$ and using the identity $d^2=0$ together with the graded Leibniz rule immediately yields
	\begin{equation}\label{2BI2}
		dF=\mu_2(dA,A)+\mu_1(dB),\qquad
		dH=\mu_2(dA,B)-\mu_2(A,dB)-\frac12\,\mu_3(dA,A,A).
	\end{equation}
	Substituting the preceding expressions for $dA$ and $dB$, and then
	using the $L_\infty$-identities \eqref{q61}, gives
	\begin{equation*}
		dF=\mu_2(F,A)+\mu_1(H),\qquad
		dH=\mu_2(F,B)-\mu_2(A,H)-\frac12\,\mu_3(F,A,A).
	\end{equation*}
	By Proposition~\ref{lem:graded-antisym-forms},
	\begin{equation*}
		\mu_2(F,A)=-\mu_2(A,F), \qquad \mu_3(F,A,A)=\mu_3(A,A,F).
	\end{equation*}
	Therefore, the identities \eqref{2BI2} are  equivalent to
	\eqref{eq:2-bianchi}.
\end{proof}

Infinitesimal gauge transformations are generated by total-degree-zero
parameters 
$\lambda_0\in\Omega^\bullet_0(M,\mathfrak L)$.
In the $2$-term case,  such parameter can be decomposed as
\begin{equation*}
	\lambda_0=a+b,\qquad
	a\in\Omega^0(M)\otimes\mathfrak L_0,\qquad
	b\in\Omega^1(M)\otimes\mathfrak L_{-1},
\end{equation*}
where $a$ is the ordinary gauge parameter and $b$ is the higher gauge parameter.
The general transformation law
\eqref{gtraA} becomes
\begin{equation*}
	\delta_{\lambda_0}\mathcal A
	=\mu_1^*(\lambda_0)+\mu_2^*(\mathcal A,\lambda_0)+\frac12\,\mu_3^*(\mathcal A,\mathcal A,\lambda_0)
	= \delta_{a,b}A+\delta_{a,b}B,
\end{equation*}
where
\begin{align}\label{eq:gauge-tra-comp}
	\delta_{a,b}A
	=da+\mu_2(A,a)-\mu_1(b),\qquad
	\delta_{a,b}B
	=db+\mu_2(A,b)+\mu_2(B,a)+\frac12\,\mu_3(A,A,a).
\end{align}
Similarly, the curvature transformation \eqref{gtraf} becomes
\begin{equation*}
	\delta_{\lambda_0}\mathcal F
	=\mu_2^*(\mathcal F,\lambda_0)+\mu_3^*(\mathcal A,\mathcal F,\lambda_0)
	=\delta_{a,b}F+\delta_{a,b}H,
\end{equation*}
where
\begin{align}\label{gauge-tra-cur}
	\delta_{a,b}F=\mu_2(F,a),\qquad
	\delta_{a,b}H=\mu_2(H,a)+\mu_2(F,b)+\mu_3(A,F,a).
\end{align}

The commutator of two infinitesimal gauge transformations satisfies
\begin{equation*}
	[\delta_{\lambda_0},\delta_{\lambda_0'}]\mathcal A
	=\delta_{\lambda_0''}\mathcal A+\mu_3^*(\mathcal F,\lambda_0,\lambda_0'),
	\qquad
	\lambda_0''=\mu_2^*(\lambda_0,\lambda_0')+\mu_3^*(\mathcal A,\lambda_0,\lambda_0').
\end{equation*}
Let $\lambda_0=a+b$ and $\lambda_0'=a'+b'$ with $a,a'\in\Omega^0(M)\otimes\mathfrak L_0$ and $b,b'\in\Omega^1(M)\otimes\mathfrak L_{-1}$. Writing $\lambda_0''=a''+b''$, a straightforward computation yields
\begin{subequations}
	\begin{align*}
		[\delta_{a,b},\delta_{a',b'}]A=\delta_{a'',b''}A,\qquad
		[\delta_{a,b},\delta_{a',b'}]B=\delta_{a'',b''}B+\mu_3^*(F,a,a').
	\end{align*}
\end{subequations}
Consequently, the gauge algebra closes off shell on the  $A$-component, while
its closure on the $B$-component is obstructed by the fake curvature $F$.   In the special case $F=0$, the algebra closes on-shell. Further details can be found in \cite{BLCM}.

\subsubsection{Four-dimensional higher Chern--Simons theory}\label{subsubsec:4dHCS}
The four-dimensional HCS theory is obtained by
specializing the homotopy MC construction to a cyclic $2$-term $L_\infty$-algebra \cite{GR,BLCM}. Let $\mathcal{A}=(A,B)$ be a $2$-connection. In four dimensions, the homotopy MC functional takes the form
\begin{equation*}
	S_4(\mathcal A)
	=\dfrac{1}{2}\,\ll \mathcal A,\mu_1^*(\mathcal A)\gg
	+\dfrac{1}{6}\,\ll \mathcal A,\mu_2^*(\mathcal A,\mathcal A)\gg
	+\dfrac{1}{24}\,\ll \mathcal A,\mu_3^*(\mathcal A,\mathcal A,\mathcal A)\gg.
\end{equation*}
We refer to  $S_4$ as the \emph{four-dimensional 2-Chern--Simons action}.

In terms of the component fields  $A$ and $B$, the action reads
\begin{align*}
	S_4(A,B)=&\int_M \Big(
	\frac12\,\langle A,dB\rangle
	+\frac12\,\langle B,dA\rangle
	+\frac12\,\langle B,\mu_1(B)\rangle +\frac16\,\langle A,\mu_2(A,B)\rangle
	-\frac16\,\langle A,\mu_2(B,A)\rangle\nonumber\\
	& 
	+\frac16\,\langle B,\mu_2(A,A)\rangle
	-\frac{1}{24}\,\langle A,\mu_3(A,A,A)\rangle
	\Big)	\nonumber \\
	=&\int_M \Big(
	\big\langle F-\dfrac{1}{2}\mu_1(B), B\big\rangle
	-\frac{1}{24}\,\langle A,\,\mu_3(A,A,A)\rangle\Big).
\end{align*}
The second expression follows from the cyclicity of the pairing and
integration by parts on the closed manifold  $M$. 
The corresponding
Lagrangian 4-form is
\begin{equation}\label{2CS-4d}
	C^4(A,B)\coloneqq \langle F-\dfrac{1}{2}\mu_1(B), B\rangle-\frac{1}{24}\,\langle A,\,\mu_3(A,A,A)\rangle,
\end{equation}
which is the \emph{four-dimensional $2$-Chern--Simons form}.

For an arbitrary variation $(\delta A,\delta B)$, the variation of the action is
\begin{equation}\label{eq:4d-variation}
	\delta S_4(A,B)=\int_M\Big(\langle F, \delta B\rangle+\langle \delta A,H\rangle\Big).
\end{equation}
It follows that the equations of motion are
\begin{equation*}
	F=0,\qquad H=0.
\end{equation*}
As expected, solutions to these equations correspond to flat 
2-connections on the trivial higher bundle over $M$, in analogy with the standard CS theory.

For each pair of gauge parameters
$
a\in\Omega^0(M)\otimes\mathfrak L_0
$
and
$
b\in\Omega^1(M)\otimes\mathfrak L_{-1}
$,
the transformation $\delta_{a,b}$  defines an infinitesimal gauge
transformations on the space of 2-connections. Substituting
\eqref{eq:gauge-tra-comp} into \eqref{eq:4d-variation}, and using the
cyclicity identities \eqref{eq:invpair-forms} together with integration
by parts, gives
\begin{equation}\label{eq:4dHCS-gauge-var-expanded}
	\begin{aligned}
		\delta_{a,b}S_4(A,B)
		=&-\int_M\Big(
		\big\langle dF+\mu_2(A,F)-\mu_1(H),\,b\big\rangle \\
		&\qquad
		+\big\langle a,\,dH+\mu_2(A,H)+\mu_2(B,F)+\tfrac12\,\mu_3(A,A,F)\big\rangle
		\Big).
	\end{aligned}
\end{equation}
By Proposition~\ref{lem:graded-antisym-forms}, we have
$\mu_2(B,F)=-\mu_2(F,B)$.
The two expressions in \eqref{eq:4dHCS-gauge-var-expanded} are
therefore precisely the 2-Bianchi identities \eqref{eq:2-bianchi}. Consequently, we obtain
\begin{equation*}
	\delta_{a,b}S_4(A,B)=0.
\end{equation*}
Thus, the four-dimensional 
2-Chern--Simons action is invariant
under infinitesimal gauge transformations off shell.

The preceding discussion concerns infinitesimal gauge transformations.
The behavior of  $S_4$ under finite higher gauge transformations is
more subtle. On manifolds with boundary, boundary contributions arise
in general, while in the semistrict case further global terms may
appear~\cite{OZ,Zucchini-2014-1}. In the strict case, the finite gauge
variation reduces to a boundary term~\cite{HC-2024}. We do not address
these global issues here and restrict attention to the infinitesimal
case.

\section{$(2n+2)$-Dimensional Higher Chern--Simons Theory}
\label{sec:2n+2-hCS}
In this section, we extend the four-dimensional HCS theory associated with cyclic 
$2$-term $L_\infty$-algebras to dimension 
$2n+2$. The corresponding construction in the strict case, where 
$\mu_3=0$, was given in~\cite{DHS-4}; here we treat the general semistrict case. Motivated by the transgression mechanism underlying ordinary CS theory in odd dimensions, we generalize the HCS theory to even dimensions and present the corresponding transgression formula and infinitesimal gauge symmetry.

\subsection{Cyclic multilinear forms on $2$-term $L_\infty$-algebras}\label{subsec:extended-pairings}
To formulate the higher-dimensional theory, we extend the notion of an invariant bilinear pairing to suitable multilinear forms on a $2$-term $L_\infty$-algebra $\mathfrak L$. Such forms have appeared previously in~\cite{WS}; we recall them here in conventions adapted to the present work. Then, we define the corresponding extended multilinear forms on $\mathfrak L$-valued differential forms.
This construction is also naturally motivated by the duality between $L_\infty$-algebra and FDA~\cite{Salgado}. Under this correspondence, invariant tensors on the FDA side are related to cyclic multilinear forms on the $L_\infty$-algebra side. In what follows, we consider cyclic forms with one $\mathfrak L_{-1}$-entry and $n$ $\mathfrak L_0$-entries.

\begin{definition}\label{def:invpoly-Pversion}
	Let $\mathfrak L=\mathfrak L_{-1}\oplus \mathfrak L_0$ be a $2$-term $L_\infty$-algebra with structure maps $\mu_1,\mu_2,$ and $\mu_3$.
	A \emph{cyclic $(n+1)$-linear form} on $\mathfrak L$ is a multilinear map
	\begin{equation*}
		\langle -,\ldots,-;\,-\rangle:\mathfrak L_0^{\times n}\times \mathfrak L_{-1}\to \mathbb{R},
	\end{equation*}
	which is symmetric in its $\mathfrak L_0$-arguments and satisfies the
	following conditions: for all $	X_1,\dots,X_n,X,X'$, $X''\in\mathfrak L_0$ and $Y,Y'\in\mathfrak L_{-1}$,
	\begin{subequations}\label{q41}
		\begin{align}
			\langle X_1,\dots,\mu_1(Y),\dots,X_n;\,Y'\rangle
			&=\langle X_1,\dots,\mu_1(Y'),\dots,X_n;\,Y\rangle,
			\label{eq:bracket-inv-mu1}\\[1mm]
			\langle X_1,\dots,X_n;\,\mu_2(X,Y)\rangle
			&=-\sum_{i=1}^n \langle X_1,\dots,\mu_2(X,X_i),\dots,X_n;\,Y\rangle,
			\label{eq:bracket-inv-mu2}\\[1mm]
			\langle X_1,\dots,X_n;\,\mu_3(X,X',X'')\rangle
			&=-\sum_{i=1}^n \langle X_1,\dots,X_{i-1},X'',X_{i+1},\dots,X_n;\,\mu_3(X,X',X_i)\rangle.
			\label{eq:bracket-inv-mu3}
		\end{align}
	\end{subequations}
\end{definition}

The identities \eqref{q41} generalize the cyclicity conditions for an invariant bilinear pairing \eqref{q14} in the special case $n=1$.  In the strict case $\mu_3=0$, condition \eqref{eq:bracket-inv-mu3} is vacuous, and the above notion reduces to the invariant multilinear forms considered in~\cite{DHS-4,DHS-5}. 

We next extend a cyclic multilinear form on  $\mathfrak L$ to $\mathfrak L$-valued differential forms. Let $M$ be a closed oriented smooth manifold.
\begin{definition}\label{def:induced-bracket-forms}
	Let $\mathfrak L=\mathfrak L_{-1}\oplus\mathfrak L_0$ be a $2$-term $L_\infty$-algebra  equipped with a cyclic $(n+1)$-linear form. For homogeneous forms
	$\omega_i\in\Omega^\bullet(M)$, $ \eta\in\Omega^\bullet(M)$
	and elements
	$X_i\in\mathfrak L_0 (1\le i\le n)$, $Y\in\mathfrak L_{-1}$, the extended multilinear form on $\Omega^\bullet(M,\mathfrak L)$  is defined by
	\begin{equation*}
		\bigl\langle \omega_1\otimes X_1,\dots,\omega_n\otimes X_n;\ \eta\otimes Y\bigr\rangle
		\coloneqq 
		\omega_1\wedge\cdots\wedge\omega_n\wedge\eta\,\langle X_1,\dots,X_n;Y\rangle.
	\end{equation*}
\end{definition}

The extended multilinear form can be viewed as a natural generalization of the extended pairing \eqref{omel}. Correspondingly, it inherits the graded symmetry and cyclicity properties of the underlying multilinear form.
\begin{proposition}\label{lem:invpoly-sym-forms}
	The extended multilinear form $\big\langle-\, ,\ldots,-;\,-\big\rangle$  on $\Omega^\bullet(M,\mathfrak L)$ is graded symmetric in its $\Omega^\bullet(M,\mathfrak L_0)$-arguments.
	More precisely, for homogeneous 
	$\omega_i,\eta\in\Omega^\bullet(M)$, $X_i\in\mathfrak L_0$ $(i=1,\dots,n)$, and
	$Y\in\mathfrak L_{-1}$,  one has
	\begin{equation*}
		\big\langle \omega_{\sigma(1)}\!\otimes\!X_{\sigma(1)}, \cdots, \omega_{\sigma(n)}\!\otimes\!X_{\sigma(n)};\ \eta\!\otimes\!Y\big\rangle
		=
		\chi(\sigma;\omega_1,\dots,\omega_n)\,
		\big\langle \omega_1\!\otimes\!X_1, \cdots, \omega_n\!\otimes\!X_n;\ \eta\!\otimes\!Y\big\rangle
	\end{equation*}
	for every $\sigma\in\mathrm{Sym}(n)$, where  $\chi(\sigma;\omega_1,\dots,\omega_n)$ is determined by
	\begin{equation*}
		\omega_{\sigma(1)}\wedge\cdots\wedge\omega_{\sigma(n)}
		=
		\chi(\sigma;\omega_1,\dots,\omega_n)\,\omega_1\wedge\cdots\wedge\omega_n.
	\end{equation*}
\end{proposition}
\begin{proof}
	The claim follows from the graded commutativity of the wedge product
	and the symmetry of the underlying multilinear form in its $\mathfrak L_0$-arguments.
\end{proof}

In the following proposition, $\mu_k$ denotes the pointwise extension of the $k$-ary structure map to $\mathfrak L$-valued differential forms, with the convention fixed in Section~\ref{subsec:higher-gauge-theory}. In particular, these extended operations are distinct from the induced operations $\mu_k^*$.
\begin{proposition}\label{lem:invpoly-forms}
	The extended multilinear form $\big\langle-\, ,\ldots,-;\,-\big\rangle$ on $\Omega^\bullet(M,\mathfrak L)$ is invariant in the following sense: for all $\omega_1,\dots,\omega_n,\eta,\omega,\omega',\omega'',\eta'\in\Omega^\bullet(M)$ and $X_1,\dots,X_n,X,X',X''\in\mathfrak L_0$, $Y,Y'\in\mathfrak L_{-1}$, 
	\begin{subequations}\label{eq:invP}
		\begin{align}
			&\big\langle \omega_1\!\otimes\!X_1, \cdots, \omega_{i-1}\!\otimes\!X_{i-1}, \mu_1(\eta\!\otimes\!Y), \omega_{i+1}\!\otimes\!X_{i+1},\cdots,\omega_n\!\otimes\!X_n; \eta'\!\otimes\!Y'\big\rangle\nonumber\\
			=&
			(-1)^{|\eta|(|\omega_{i+1}|+\cdots+|\omega_{n}| )+|\eta'|(|\eta|+|\omega_{i+1}|+\cdots+|\omega_{n}| )}
			\big\langle \omega_1\!\otimes\!X_1,\cdots,\mu_1(\eta'\!\otimes\!Y'),\cdots, \omega_n\!\otimes\!X_n;\eta\!\otimes\!Y\big\rangle,
			\label{eq:invP-mu1-forms}\\[1mm]
			&	\big\langle \omega_1\!\otimes\!X_1,\cdots, \omega_n\!\otimes\!X_n;\mu_2(\omega\!\otimes\!X,\eta\!\otimes\!Y)\big\rangle\nonumber\\
			=&
			-\sum_{i=1}^n
			(-1)^{|\omega|\,(|\omega_i|+\cdots+|\omega_{n}|)}
			\big\langle \omega_1\!\otimes\!X_1,\cdots,\mu_2(\omega\!\otimes\!X,\omega_i\!\otimes\!X_i), \cdots,\omega_n\!\otimes\!X_n;\eta\!\otimes\!Y\big\rangle,
			\label{eq:invP-mu2-forms}\\[1mm]
			&\big\langle \omega_1\!\otimes\!X_1,\cdots, \omega_n\!\otimes\!X_n; \mu_3(\omega\!\otimes\!X,\omega'\!\otimes\!X',\omega''\!\otimes\!X'')\big\rangle\nonumber\\
			=&
			-\sum_{i=1}^n
			(-1)^{|\omega''|\,(|\omega_i|+\cdots+|\omega_{n}|+|\omega|+|\omega'|)+ |\omega_i|\,(|\omega_{i+1}|+\cdots+|\omega_{n}|+|\omega|+|\omega'|)}\,
			\big\langle \omega_1\!\otimes\!X_1,\cdots,\omega_{i-1}\!\otimes\!X_{i-1}, \omega''\!\otimes\!X'',\nonumber\\
			& \omega_{i+1}\!\otimes\!X_{i+1},\cdots, \omega_n\!\otimes\!X_n\mu_3(\omega\!\otimes\!X,\omega'\!\otimes\!X',\omega_i\!\otimes\!X_i)\big\rangle.
			\label{eq:invP-mu3-forms}
		\end{align}
	\end{subequations}
\end{proposition}
\begin{proof}
	Each identity follows by applying the corresponding cyclicity condition \eqref{q41}  to the $\mathfrak L$-components and then reordering the differential-form factors into the prescribed order. The resulting Koszul signs are precisely those displayed in \eqref{eq:invP}. 
\end{proof}

\subsection{Higher Chern--Simons forms in $2n+2$ dimensions}\label{subsec4.2}
We now generalize the four-dimensional HCS form of
Section~\ref{subsubsec:4dHCS} to arbitrary even dimension $2n+2$. Our construction is motivated by the Chern--Simons--Antoniadis--Savvidy
forms, namely, the extended CS forms  constructed from
non-abelian tensor gauge fields~\cite{FIP}. In the strict case $\mu_3=0$, the corresponding construction for Lie 2-algebra gauge
theory was obtained in~\cite{DHS-4}. The purpose of this section is to
extend that construction to semistrict $2$-term $L_\infty$-algebras.

Let $(A,B)$ be a $2$-connection with values in the  $2$-term $L_\infty$-algebra $\mathfrak L=\mathfrak L_{-1}\oplus\mathfrak L_0$. For a homogeneous differential form $U$ with values in $\mathfrak{L}_{-1}$ or $\mathfrak{L}_{0}$, we define a covariant derivative by
\begin{equation}\label{eq:covD}
	DU \coloneqq dU+\mu_2(A,U).
\end{equation}
By the cyclicity condition satisfied by the operation $\langle -,\dots,-;\,-\rangle$ on $\mathfrak{L}$, we have the following result.
\begin{lemma}\label{thm:d-pairing}
	Let $W_i\in\Omega^{k_i}(M)\otimes\mathfrak L_0$ for $i=1, \cdots, n$ and
	$V\in\Omega^{t}(M)\otimes\mathfrak L_{-1}$. Then 
	\begin{equation}\label{eq:d-pairing}
		\begin{aligned}
			d\langle W_1, \cdots, W_n;\,V\rangle
			&=\sum_{i=1}^{n}(-1)^{k_1+\cdots+k_{i-1}}
			\langle W_1,\cdots, DW_i,\cdots, W_n;\,V\rangle\\
			&\quad+(-1)^{k_1+\cdots+k_n}\langle W_1,\cdots, W_n;\,DV\rangle .
		\end{aligned}
	\end{equation}
\end{lemma}
\begin{proof}
	Substituting \eqref{eq:covD} into the right-hand side of
	\eqref{eq:d-pairing}, the terms containing $dW_i$ and $dV$ reproduce
	the graded Leibniz rule for the exterior derivative. The remaining
	terms are
	\begin{equation*}
		\sum_{i=1}^{n}(-1)^{k_1+\cdots+k_{i-1}}
		\langle W_1,\dots, \mu_2(A,W_i),\dots, W_n;\,V\rangle
		+(-1)^{k_1+\cdots+k_n}\langle W_1,\dots, W_n;\,\mu_2(A,V)\rangle.
	\end{equation*}
	They vanish by the graded cyclicity identity
	\eqref{eq:invP-mu2-forms}.
\end{proof}

Then we denote the corresponding $2$-curvature $(F,H)$ by
\begin{align*}
	F= DA-\frac12\,\mu_2(A,A)+\mu_1(B),\qquad
	H= DB-\frac16\,\mu_3(A,A,A),
\end{align*}
and the resulting $2$-Bianchi identities by
\begin{align}\label{2-bianchi}
	DF=\mu_1(H),\qquad
	DH=\mu_2(F,B)-\frac12\,\mu_3(A,A,F).
\end{align}
For infinitesimal gauge parameters $a\in\Omega^0(M)\otimes\mathfrak L_0$ and $b\in\Omega^1(M)\otimes\mathfrak L_{-1}$, the gauge transformations can be written as
\begin{align}\label{gauge-tra-comp}
	\delta_{a,b}A
	=Da-\mu_1(b),\qquad
	\delta_{a,b}B
	=Db+\mu_2(B,a)+\frac12\,\mu_3(A,A,a).
\end{align}

We now turn to the higher analogue of the Pontryagin--Chern form, which is the main subject in this subsection.
Extended characteristic forms in non-abelian tensor gauge theory were constructed in~\cite{GS,IAGS,GGGS}; moreover, gauge-invariant, closed, and metric-independent 
$(2n+3)$-forms were exhibited in~\cite{FIP}, and the strict Lie $2$-algebra counterpart was given in~\cite{DHS-4}. Motivated by these results, we define the  \emph{higher Pontryagin--Chern form} associated with a cyclic 
$(n+1)$-linear form by
\begin{equation}\label{eq:higher-PC}
	\Gamma_{2n+3}(F,H)
	\coloneqq \bigl\langle \underbrace{F,\ldots,F}_{n\ \text{copies}};\,H\bigr\rangle \;=\; \langle F^n;H\rangle.
\end{equation}
Since $F$ and $H$ have form of degrees 2 and 3, respectively, $\Gamma_{2n+3}(F,H)$ is a differential form of degree $2n+3$. 

By Lemma \ref{thm:d-pairing}, we have the following result.
\begin{proposition}\label{prop:PC-closed}
	For every $n\ge 1$, the higher Pontryagin--Chern form $\Gamma_{2n+3}(F,H)$ defined in \eqref{eq:higher-PC}
	is closed, namely,
	\begin{equation*}
		d\Gamma_{2n+3}(F,H)=0.
	\end{equation*}
\end{proposition}
\begin{proof}
	By Lemma~\ref{thm:d-pairing},
	\begin{equation*}
		d\langle F^n;H\rangle
		=
		n\langle DF,F^{n-1};H\rangle
		+
		\langle F^n;DH\rangle.
	\end{equation*}
	Using \eqref{2-bianchi}, this becomes
	\begin{align*}
		d\langle F^n;H\rangle
		={}&
		n\langle\mu_1(H),F^{n-1};H\rangle
		+
		\langle F^n;\mu_2(F,B)\rangle-
		\frac{1}{2}
		\langle F^n;\mu_3(A,A,F)\rangle.
	\end{align*}
	The first term vanishes by the graded $\mu_1$-cyclicity relation, the
	second by $\mu_2$-cyclicity together with $\mu_2(F, F)=0$, and the
	third by the $\mu_3$-cyclicity relation. Hence,
	$d\,\Gamma_{2n+3}(F,H)=0$.
\end{proof}

Analogously to the ordinary and extended Pontryagin--Chern forms, the higher analogue  $\Gamma_{2n+3}(F,H)$ enjoys the following gauge invariance property.
\begin{proposition}\label{prop:PC-invariant}
	For every $n\ge 1$,  the form $\Gamma_{2n+3}(F,H)$ is invariant under the infinitesimal gauge
	transformations \eqref{gauge-tra-comp},  i.e.,
	\begin{equation*}
		\delta_{a,b}\Gamma_{2n+3}(F,H)=0.
	\end{equation*}
\end{proposition}
\begin{proof}
	Using \eqref{gauge-tra-cur}, we have
	\begin{equation*}
		\delta_{a,b}\langle F^n;H\rangle
		=n\bigl\langle \mu_2(F,a), F^{n-1};H\bigr\rangle
		+\bigl\langle F^n;\mu_2(H,a)\bigr\rangle
		+\bigl\langle F^n;\mu_2(F,b)\bigr\rangle
		+\bigl\langle F^n;\mu_3(A,F,a)\bigr\rangle.
	\end{equation*}
	The first two terms cancel by $\mu_2$-cyclicity \eqref{eq:invP-mu2-forms}. Moreover,
	\begin{equation*}
		\langle F^n;\mu_2(F,b)\rangle
		=
		-n\langle\mu_2(F,F),F^{n-1};b\rangle
		=
		0,
	\end{equation*}
	where the last equality follows from graded antisymmetry. Finally, the
	last term vanishes by the $\mu_3$-cyclicity identity \eqref{eq:invP-mu3-forms}. Thus we have
	$\delta_{a,b}\langle F^n;H\rangle=0$.
\end{proof}

Since $\Gamma_{2n+3}(F,H)$ is closed, the Poincar\'e lemma implies that it is locally exact. Thus, locally there exists a $(2n+2)$-form $\mathfrak C^{2n+2}(A,B)$ such that 
\begin{equation*}
	\Gamma_{2n+3}(F,H)=d\,\mathfrak{C}^{2n+2}(A,B).
\end{equation*}
We will refer to such a form $\mathfrak{C}^{2n+2}(A,B)$ as  a $(2n+2)$-dimensional HCS form associated with the 2-term $L_{\infty}$-algebra $\mathfrak{L}$. We now turn to give the explicit formula for $\mathfrak{C}^{2n+2}(A,B)$.
To this end, consider arbitrary variations $\delta A$ and $\delta B$, and the corresponding curvature variations are given by
\begin{align*}
	\delta F=D(\delta A)+\mu_1(\delta B),\qquad
	\delta H=D(\delta B)+\mu_2(\delta A,B)-\frac12\,\mu_3(\delta A,A,A).
\end{align*}
Then we have
\begin{align*}
	\delta \Gamma_{2n+3}(F,H)
	=n\langle D(\delta A)+\mu_1(\delta B),F^{n-1};H\rangle
	+\big\langle F^{n};D(\delta B)+\mu_2(\delta A,B)-\frac12\,\mu_3(\delta A,A,A)\big\rangle.
\end{align*}
Using Lemma~\ref{thm:d-pairing}, the Bianchi identities \eqref{2-bianchi}, and the cyclicity relations, we obtain
\begin{align*}
	\langle D(\delta A), F^{n-1}; H\rangle
	&=d\langle \delta A, F^{n-1}; H\rangle
	+\big\langle \delta A, F^{n-1}; \mu_2(F,B)-\frac{1}{2}\mu_3(F,A,A)\big\rangle, \\
	\langle F^n; D(\delta B)\rangle
	&=d\langle F^n; \delta B\rangle -n \langle \mu_1(H), F^{n-1};\delta B\rangle.
\end{align*}
Hence, we get the infinitesimal transgression formula
\begin{align}\label{delP}
	\delta \Gamma_{2n+3}(F,H)
	=d\big(n\langle \delta A, F^{n-1};H\rangle+\langle F^{n};\delta B\rangle\big).
\end{align}

Following \cite{BZ}, we introduce a one-parameter family of higher potentials and curvatures parametrized by $t\in[0,1]$,
\begin{align*}
	A_t&\coloneqq tA,
	\qquad
	F_t\coloneqq tF+\frac{t^2-t}{2}\,\mu_2(A,A),\\
	B_t&\coloneqq tB,
	\qquad
	H_t\coloneqq tH+(t^2-t)\,\mu_2(A,B)+\frac{t-t^3}{6}\,\mu_3(A,A,A).
\end{align*}
Then, the infinitesimal variation $\delta=\delta t\,(\partial/\partial t)$  satisfies $\delta A_t=\delta t\,A$ and
$\delta B_t=\delta t\,B$. Substituting them into~\eqref{delP} and integrating along the path $t\in[0,1]$ yields 
\begin{equation*}
	\Gamma_{2n+3}(F,H)=\langle F^{n};H\rangle=d\mathfrak{C}^{2n+2}(A,B),
\end{equation*}
where 
\begin{equation}\label{2CS}
	\mathfrak{C}^{2n+2}(A,B)
	=\int_{0}^{1}\!dt\,
	\Big(\, n\,\langle A, F_t^{\,n-1};H_t\rangle+\langle F_t^{\,n};B\rangle \Big).
\end{equation}
When $\mu_3=0$, \eqref{2CS} reduces to the HCS form of the strict Lie 2-algebra theory
constructed in~\cite{DHS-4}.

For $n=1$, \eqref{2CS} reduces to
\begin{align*} 
	\mathfrak C^4(A,B) =& \int_0^1dt\, \Big( \langle A;H_t\rangle+\langle F_t;B\rangle \Big) \nonumber\\ 
	=& \langle A; \frac{1}{2}dB+\frac{1}{3}\mu_2(A,B) -\frac{1}{24}\mu_3(A,A,A) \rangle + \langle \frac{1}{2}dA+\frac{1}{6}\mu_2(A,A) +\frac{1}{2}\mu_1(B); B \rangle.
\end{align*} 
Using cyclicity and the graded Leibniz rule, we see that $\mathfrak C^4(A,B)$ differs from the four-dimensional $2$-Chern--Simons form \eqref{2CS-4d} by an exact form. In particular, on a closed four-manifold, we have
\begin{equation*}
	\int_M\mathfrak C^4(A,B) = \int_M\Big( \langle F-\frac{1}{2}\mu_1(B);B \rangle -\frac{1}{24} \langle A;\mu_3(A,A,A)\rangle \Big),
\end{equation*}
which is precisely the four-dimensional $2$-Chern--Simons action associated with a $2$-term $L_\infty$-algebra \cite{GR,OZ,Zucchini-2014-1}.

\subsection{Higher transgression forms in $2n+2$ dimensions}
In this section, we present the transgression construction associated with the higher Pontryagin--Chern form for $2$-term $L_\infty$-algebras.  We first generalize the Chern--Weiltheorem and construct a gauge-invariant $(2n+2)$-dimensional transgression form. We then examine its relation to the HCS forms constructed in the previous subsection \ref{subsec4.2}. Finally, we turn to a discussion of transgression forms as Lagrangians for gauge field theories.

\subsubsection{Higher Chern--Weil theorem}\label{HCW}
Let $(A_0,B_0)$ and $(A_1,B_1)$ be two $2$-connections with curvatures $(F_0,H_0)$ and $(F_1,H_1)$, respectively. Thus, for $i=0,1$, we have
\begin{align*}
	F_i=dA_i+\frac{1}{2}\mu_2(A_i,A_i)+\mu_1(B_i),\qquad
	H_i=dB_i+\mu_2(A_i,B_i)-\frac{1}{6}\mu_3(A_i,A_i,A_i).
\end{align*}
Set
$\Theta=A_1-A_0$ and $\Phi=B_1-B_0$, and consider the family of interpolating fields
\begin{equation}\label{inter-field}
	A_t\coloneqq A_0+t\Theta,\qquad B_t\coloneqq B_0+t\Phi,
\end{equation}
with $t\in[0,1]$.  Their corresponding curvatures then take the form
\begin{align*}
	F_t
	&=F_0+t\big(d\Theta+\mu_2(A_0,\Theta)+\mu_1(\Phi)\big)
	+\frac{t^2}{2}\mu_2(\Theta,\Theta),\\[0.2em]
	H_t
	&=H_0+t\big(d\Phi+\mu_2(A_0,\Phi)+\mu_2(\Theta,B_0)
	-\frac{1}{2}\mu_3(A_0,A_0,\Theta)\big)
	+t^2\big(\mu_2(\Theta,\Phi)-\frac{1}{2}\mu_3(A_0,\Theta,\Theta)\big)\nonumber\\
	&\quad
	-\frac{t^3}{6}\mu_3(\Theta,\Theta,\Theta).
\end{align*}

Let $D_t$ denote the covariant derivative associated with $A_t$,  defined by
\begin{equation*} 
	D_tU\coloneqq dU+\mu_2(A_t,U)
\end{equation*} 
for $U\in \Omega^{\bullet}(M, \mathfrak{L}_{-1} \ \text{or} \ \mathfrak{L}_0)$.
Then a direct computation then yields 
\begin{subequations}\label{q71}
	\begin{align*}
		\frac{\partial F_t}{\partial t}=D_t\Theta+\mu_1(\Phi),\qquad
		\frac{\partial H_t}{\partial t}=D_t\Phi+\mu_2(\Theta,B_t)
		-\frac{1}{2}\mu_3(A_t,A_t,\Theta).
	\end{align*}
\end{subequations}

Motivated by the classical and extended Chern--Weil theorems \cite{FIERPS,Salgado}, we establish a higher Chern--Weil transgression formula for $2$-term $L_\infty$-algebras. This formula shows that the difference of the $(2n+3)$-dimensional characteristic forms is the exterior derivative of a $(2n+2)$-dimensional higher transgression form.
\begin{theorem}\label{HCW-t}
	Let 
	$\Gamma_{2n+3}(F,H)$
	be the higher Pontryagin--Chern form defined by
	\eqref{eq:higher-PC}. Then
	\begin{equation}\label{eq:higher-Chern-Weil}
		\Gamma_{2n+3}(F_1,H_1)-\Gamma_{2n+3}(F_0,H_0)
		=
		d\,T^{2n+2}(A_0,B_0;A_1,B_1),
	\end{equation}
	where
	\begin{equation}\label{eq:higher-transgression}
		T^{2n+2}(A_0,B_0;A_1,B_1)
		\coloneqq 
		\int_0^1dt\,
		\Bigl(
		n\langle\Theta,F_t^{n-1};H_t\rangle
		+
		\langle F_t^n;\Phi\rangle
		\Bigr)
	\end{equation}
	is called a \emph{higher transgression form} associated with $(A_0,B_0)$ and $(A_1,B_1)$.
\end{theorem}
\begin{proof}
	First by the fact that $(F_t, H_t)$ interpolates between $(F_0, H_0)$ and $(F_1, H_1)$, we have
	\begin{equation*}
		\Gamma_{2n+3}(F_1,H_1)-\Gamma_{2n+3}(F_0,H_0)
		=\langle F_1^{\,n};H_1\rangle-\langle F_0^{\,n};H_0\rangle
		=\int_0^1\!dt\,\frac{\partial}{\partial t}\,\langle F_t^{\,n};H_t\rangle .
	\end{equation*}
	By \eqref{q71}, we obtain
	\begin{align}
		\frac{\partial}{\partial t}\,\langle F_t^{\,n};H_t\rangle
		&=n\big\langle F_t^{\,n-1},\frac{\partial F_t}{\partial t};H_t\big\rangle
		+\big\langle F_t^{\,n};\frac{\partial H_t}{\partial t}\big\rangle\nonumber\\
		&=n\langle F_t^{\,n-1},D_t\Theta+\mu_1(\Phi);H_t\rangle
		+\langle F_t^{\,n};D_t\Phi+\mu_2(\Theta,B_t)-\tfrac12\mu_3(A_t,A_t,\Theta)\rangle.
		\label{eq:HCW-proof-2}
	\end{align}
	
	It follows from \eqref{eq:d-pairing} that
	\begin{equation}\label{ewrq}
		d\langle\Theta,F_t^{n-1};H_t\rangle
		=
		\langle D_t\Theta,F_t^{n-1};H_t\rangle
		-(n-1)
		\langle\Theta,D_tF_t,F_t^{n-2};H_t\rangle
		-\langle\Theta,F_t^{n-1};D_tH_t\rangle.
	\end{equation}
	Using the 2-Bianchi identities
	\begin{equation*}
		D_t F_t=\mu_1(H_t),\qquad D_tH_t=\mu_2(F_t, B_t)-\dfrac{1}{2}\mu_3(A_t, A_t, F_t),
	\end{equation*}
	we find that the term involving $\mu_1(H_t)$ in \eqref{ewrq} vanishes by $\mu_1$-cyclicity and the odd form degree of $H_t$. Hence
	\begin{align*}
		d\langle\Theta,F_t^{n-1};H_t\rangle
		=\langle F_t^{n-1},D_t\Theta;H_t\rangle-
		\langle\Theta,F_t^{n-1};\mu_2(F_t,B_t)\rangle+
		\frac{1}{2}
		\langle\Theta,F_t^{n-1};
		\mu_3(F_t,A_t,A_t)\rangle.
	\end{align*}
	The cyclicity identities \eqref{eq:invP-mu2-forms} and \eqref{eq:invP-mu3-forms} imply
	\begin{align*}
		n\langle \Theta, F_t^{\,n-1};\mu_2(F_t,B_t)\rangle
		&=-\langle F^n_t;\mu_2(\Theta, B_t)\rangle, \\
		n\,\langle \Theta, F_t^{\,n-1};\mu_3(F_t,A_t,A_t)\rangle
		&=-\langle F_t^{\,n};\mu_3(A_t,A_t,\Theta)\rangle. 
	\end{align*}
	Consequently,
	\begin{equation}\label{eq:HCW-proof-6}
		n\langle F_t^{\,n-1},D_t\Theta;H_t\rangle
		=
		n\,d\langle \Theta, F_t^{\,n-1};H_t\rangle
		-n\langle F^n_t;\mu_2(\Theta, B_t)\rangle
		+\frac12\langle F_t^{\,n};\mu_3(A_t,A_t,\Theta)\rangle .
	\end{equation}
	
	Similarly, applying the invariance property once more to the second term of \eqref{eq:HCW-proof-2} gives
	\begin{equation}\label{eq:HCW-proof-7}
		\langle F_t^{\,n};D_t\Phi\rangle
		=d\langle F_t^{\,n};\Phi\rangle
		-n\langle F_t^{\,n-1},\mu_1(\Phi);H_t\rangle.
	\end{equation}
	
	Finally, inserting \eqref{eq:HCW-proof-6} and \eqref{eq:HCW-proof-7} into \eqref{eq:HCW-proof-2}, all non-exact terms cancel, and yields
	\begin{equation*}
		\frac{\partial}{\partial t}\,\langle F_t^{\,n};H_t\rangle
		=d\Big(n\langle \Theta, F_t^{\,n-1};H_t\rangle+\langle F_t^{\,n};\Phi\rangle\Big).
	\end{equation*}
	Thus, integration over $t\in [0, 1]$ proves \eqref{eq:higher-Chern-Weil}.
\end{proof}

The HCS form  \eqref{2CS} can be recovered by taking
$(A_0,B_0)=(0,0)$ and $(A_1,B_1)=(A,B)$.
In this case, 
$\Theta=A$ and $ \Phi =B$,
and the higher transgression form \eqref{eq:higher-transgression} reduces to
\begin{equation}\label{T-HCS}
	T^{2n+2}(0,0;A,B)
	=
	\int_0^1dt\,
	\Bigl(
	n\langle A,F_t^{n-1};H_t\rangle
	+
	\langle F_t^n;B\rangle
	\Bigr)
	=
	\mathfrak C^{2n+2}(A,B).
\end{equation}
Thus, the $(2n+2)$-dimensional  HCS form  is a
transgression form associated with the zero reference 2-connection.
For this specific case, the Chern--Weiltheorem implies
\begin{equation*}
	d \mathfrak C^{2n+2}(A,B)=\Gamma_{2n+3}(F,H).
\end{equation*}
Since $\Gamma_{2n+3}(F,H)$ is gauge invariant, it follows that 
$d \delta \mathfrak C^{2n+2}(A,B)=0$. This means that the gauge variation $\delta\mathfrak C^{2n+2}(A,B)$ is locally exact: it can be written as $d\Omega^{2n+1}$ for some $(2n+1)$-form $\Omega^{2n+1}$. However, a 2-connection cannot be set to zero globally unless the underlying higher bundle is topologically trivial \cite{RZ-AKSZ}. Consequently, the HCS form is only locally defined. By contrast, the higher transgression form is globally well-defined in principle. As we will see, it is also invariant under gauge transformations. These features closely parallel the familiar properties of ordinary CS theory.


Finally, we show that when $\mu_3=0$, the higher transgression form \eqref{eq:higher-transgression} reduces precisely to the strict 2-transgression form constructed in~\cite{DHS-4}. Moreover, its local expression possesses the same formal structure as the Antoniadis--Savvidy transgression form given in~\cite{FIP}. The present construction therefore provides a semistrict generalization of these earlier results to cyclic 2-term $L_\infty$-algebras.

\subsubsection{Higher transgression  as a lagrangian}
Building on the transgression gauge theories introduced in~\cite{FIERPS},
we now construct a higher gauge theory whose lagrangian is given by a higher
transgression form. While the special case of strict Lie $2$-algebras was studied in~\cite{DHS-4}, the present work extends the construction to the more general setting of semistrict $2$-term $L_\infty$-algebras.

Let $M$ be an oriented manifold of dimension $2n+2$, possibly with
boundary. Given two endpoint 2-connections $(A_0, B_0)$ and $(A_1, B_1)$, we define the higher transgression action by
\begin{align}\label{stre-ac}
	S^{2n+2}_{\mathrm T}(A_0,B_0;A_1,B_1)
	\coloneqq \int_M T^{2n+2}(A_0,B_0;A_1,B_1) = \int_M \int_0^1\!dt\,
	\Big( n\,\langle \Theta, F_t^{\,n-1};H_t\rangle+\langle F_t^{\,n};\Phi\rangle\Big).
\end{align}
The two endpoint 2-connections enter the action symmetrically.
More precisely, reversing the interpolation parameter $t \longrightarrow 1-t$ yields
\begin{equation*}
	S^{2n+2}_{\text{T}}(A_0,B_0;A_1,B_1)=-	S^{2n+2}_{\text{T}}(A_1,B_1;A_0,B_0).
\end{equation*}

We now derive the field equations by varying the transgression Lagrangian. The computation is lengthy but straightforward; we therefore highlight only the main steps. We first treat the endpoint 2-connections as independent dynamical fields and consider the infinitesimal variations
\begin{equation*}
	A_i\longmapsto A_i+\delta A_i,
	\qquad
	B_i\longmapsto B_i+\delta B_i,
	\qquad i=0,1.
\end{equation*}
By \eqref{inter-field},
we have
\begin{align*}
	\delta A_t&=\delta A_0+t \delta \Theta, \qquad \delta\Theta=\delta A_1-\delta A_0,  \\
	\delta B_t&=\delta B_0+t\delta \Phi, \qquad \delta\Phi =\delta B_1-\delta B_0.
\end{align*}
The variations of the corresponding curvatures are given by
\begin{align}\label{q31}
	\delta F_t
	=D_t\delta A_t+\mu_1(\delta B_t),
	\qquad
	\delta H_t
	=D_t\delta B_t+\mu_2(\delta A_t,B_t)
	-\frac12\mu_3(A_t,A_t,\delta A_t).
\end{align}

Substituting \eqref{q31}  into the variation of the transgression
form gives
\begin{align}\label{eq01}
	&\delta T^{2n+2}(A_0,B_0;A_1,B_1)\nonumber\\
	=&\int_{0}^{1}dt \Big( n\langle \delta \Theta, F^{n-1}_t;H_t\rangle +n(n-1)\langle \Theta, D_t \delta A_t +\mu_1(\delta B_t), F^{n-2}_t; H_t\rangle +n\langle \Theta, F^{n-1}_t; D_t\delta B_t  \nonumber\\
	&+ \mu_2(\delta A_t, B_t)-\dfrac{1}{2}\mu_3(A_t, A_t, \delta A_t)\rangle+n\langle D_t \delta A_t +\mu_1(\delta B_t), F^{n-1}_t;\Phi\rangle +\langle F^n_t;\delta \Phi\rangle
	\Big).
\end{align}
Using Lemma~\ref{thm:d-pairing}, the graded Leibniz rule for $D_t$,
the higher Bianchi identities, and the cyclicity properties of the
pairing,  we have the following identities:
	\begin{align*}
		\langle \Theta,D_t\delta A_t,F_t^{n-2};H_t\rangle
		=&
		\langle D_t\Theta,\delta A_t,F_t^{n-2};H_t\rangle+
		\langle\Theta,\delta A_t,F_t^{n-2};
		\mu_2(F_t,B_t)-\frac12\mu_3(A_t,A_t,F_t)\rangle
		\nonumber\\
		&-
		d\langle\Theta,\delta A_t,F_t^{n-2};H_t\rangle,
		\\
		\langle\Theta,F_t^{n-1};D_t\delta B_t\rangle
		=&
		\langle D_t\Theta,F_t^{n-1};\delta B_t\rangle
		-(n-1)\langle\Theta,\mu_1(H_t),F_t^{n-2};\delta B_t\rangle
		-
		d\langle\Theta,F_t^{n-1};\delta B_t\rangle,
		\\
		\langle D_t\delta A_t,F_t^{n-1};\Phi\rangle
		=&
		(n-1)\langle\delta A_t,\mu_1(H_t),F_t^{n-2};\Phi\rangle+
		\langle\delta A_t,F_t^{n-1};D_t\Phi\rangle
		+d\langle\delta A_t,F_t^{n-1};\Phi\rangle.
	\end{align*}
Thus, the remaining $\mu_2$- and $\mu_3$-dependent terms in \eqref{eq01} are rearranged by
using the invariance identities
\eqref{eq:invP-mu2-forms} and \eqref{eq:invP-mu3-forms}:
	\begin{align*}
		\langle \Theta, \delta A_t, F^{n-2}_t; \mu_2(F_t, B_t)\rangle =&-\langle \mu_2(F_t, \Theta), \delta A_t, F^{n-2}_t;  B_t\rangle-\langle \Theta, \mu_2(F_t, \delta A_t), F^{n-2}_t; B_t\rangle,\\
		(n-1)\langle \Theta, \delta A_t, F^{n-2}_t; \mu_3(A_t, A_t, F_t)\rangle=&  \langle \delta A_t, F^{n-1}_t; \mu_3(A_t, A_t, \Theta)\rangle-\langle \Theta, F^{n-1}_t; \mu_3(A_t, A_t, \delta A_t)\rangle,\\
		\langle \delta A_t, F^{n-1}_t; \mu_2(\Theta, B_t)\rangle=&\langle \mu_2(\Theta, \delta A_t), F^{n-1}_t; B_t\rangle -(n-1)\langle \delta A_t, \mu_2(\Theta, F_t), F^{n-2}_t; B_t\rangle, \\
		\langle \Theta, F^{n-1}_t;\mu_2(\delta A_t, B_t)\rangle=&\langle \mu_2(\delta A_t, \Theta), F^{n-1}_t; B_t\rangle-(n-1)\langle \Theta, \mu_2(\delta A_t, F_t), F^{n-2}_t; B_t\rangle.
	\end{align*}
	Substituting the above identities into \eqref{eq01} and simplifying terms, we obtain
		\begin{align}\label{TTT1}
			\delta T^{2n+2}(A_0,B_0;A_1,B_1)
			&=\int_{0}^{1}\!\!dt \Big(n\langle \delta \Theta, F^{n-1}_t; H_t\rangle 
			+n(n-1)\langle \delta A_t, D_t \Theta+\mu_1(\Phi), F^{n-2}_t; H_t\rangle \nonumber\\
			&\qquad+n\langle \delta A_t, F^{n-1}_t; D_t \Phi +\mu_2(\Theta, B_t)
			-\tfrac12\mu_3(A_t, A_t, \Theta)\rangle \nonumber\\
			&\qquad+n\langle D_t \Theta+\mu_1(\Phi), F^{n-1}_t;\delta B_t\rangle 
			+\langle F^n_t, \delta \Phi\rangle\Big)-d\hat{\Pi},
		\end{align}
	where 
	\begin{equation}\label{Pi-hat}
		\hat{\Pi}= n  \int_{0}^{1} dt \Big((n-1) \langle \Theta, \delta A_t, F^{n-2}_t;H_t \rangle 
		+ \langle \Theta, F^{n-1}_t; \delta B_t \rangle 
		+ \langle \delta A_t, F^{n-1}_t; \Phi \rangle\Big).
	\end{equation}
	
	The bulk integrand in \eqref{TTT1} is observed to be a total derivative with respect to the interpolation parameter $t$. A direct computation then yields
	\begin{align*}
		\dfrac{\partial}{\partial t}\delta A_t&=\delta \Theta,	\qquad \dfrac{\partial}{\partial t}F_t=D_t\Theta +\mu_1(\Phi),\\
		\dfrac{\partial}{\partial t}\delta B_t&=\delta \Phi,\qquad  \dfrac{\partial}{\partial t}H_t=D_t\Phi + \mu_2(\Theta, B_t)-\dfrac{1}{2}\mu_3(A_t, A_t, \Theta).
	\end{align*}
	Substituting these relations into the derivatives of the following two terms yields
	\begin{subequations}
		\begin{align*}
			\dfrac{\partial}{\partial t}\langle \delta A_t, F^{n-1}_t; H_t\rangle =&\langle \delta\Theta, F^{n-1}_t; H_t\rangle+(n-1)\langle\delta A_t,  D_t \Theta+\mu_1(\Phi),F^{n-2}_t; H_t\rangle\nonumber\\ 
			&+\langle \delta A_t, F^{n-1}_t;D_t \Phi+\mu_2(\Theta, B_t)-\dfrac{1}{2}\mu_3(A_t, A_t, \Theta) \rangle, \\ 
			\dfrac{\partial}{\partial t}\langle F^n_t; \delta B_t\rangle =&n\langle D_t \Theta +\mu_1(\Phi),F^{n-1}_t; \delta B_t\rangle +\langle F^{n}_t; \delta \Phi \rangle.
		\end{align*}
	\end{subequations}
	Consequently, the variation of the transgression form admits the compact expression
	\begin{align*}
		\delta T^{2n+2}(A_0, B_0;A_1, B_1)
		=&\int_{0}^{1}dt \Big(n \dfrac{\partial}{\partial t}\langle \delta A_t, F^{n-1}_t; H_t\rangle + \dfrac{\partial}{\partial t}\langle F^n_t; \delta B_t\rangle \Big)- d \hat{\Pi}\\
		=&n \langle \delta A_1, F^{n-1}_1; H_1\rangle + \langle F^n_1;\delta B_1\rangle - n \langle \delta A_0, F^{n-1}_0; H_0\rangle- \langle F^n_0; \delta B_0\rangle - d \hat{\Pi}.
	\end{align*}
	
	Accordingly, the variation of the transgression action is
		\begin{align}
			\delta S^{2n+2}_{\mathrm T}(A_0,B_0;A_1,B_1)
			=&\int_M \delta T^{2n+2}(A_0,B_0;A_1,B_1)\nonumber\\
			=&\int_M\Big( n \langle \delta A_1, F^{n-1}_1; H_1\rangle + \langle F^n_1;\delta B_1\rangle- n \langle \delta A_0, F^{n-1}_0; H_0\rangle\nonumber\\
			& - \langle F^n_0; \delta B_0\rangle \Big)-\int_{\partial M}  \hat{\Pi}.\label{eq-TT}
		\end{align}
	This result (but for the precise form of the boundary term $ \hat{\Pi}$) can be intuitively anticipated directly from the higher Chern--Weiltheorem.  An explicit computation confirms this expectation and yields equation \eqref{Pi-hat}.
	Then the field equations are given by
	\begin{align*}
		\langle X, F^{n-1}_i; H_i\rangle &=0, \qquad 
		\forall\, X\in \mathfrak{L}_0,\\
		\langle F^{n}_i; Y\rangle&=0, \qquad 
		\forall\, Y\in \mathfrak{L}_{-1},
	\end{align*}
	where $i=0,1$. 
	The boundary conditions are obtained by requiring the vanishing of $ \hat{\Pi}$ on $\partial M$:
	\begin{equation*}
		\int_{0}^{1} dt \Big((n-1) \langle \Theta, \delta A_t, F^{n-2}_t;H_t \rangle 
		+ \langle \Theta, F^{n-1}_t; \delta B_t \rangle 
		+ \langle \delta A_t, F^{n-1}_t; \Phi \rangle\Big)\Big|_{\partial M}=0.
	\end{equation*}
	Thus, two independent HCS theories defined on the manifold $M$ become inextricably coupled at the boundary.  
	This feature is identical to the corresponding features in ordinary transgression theories \cite{FIERPS} and strict HCS theories \cite{DHS-4}.

	We now turn to the gauge invariance of the higher transgression action, which is the subject of the following proposition.
	\begin{proposition}\label{prop:gauge-inv-closed}
		Let $M$ be a manifold without boundary. Then the transgression action $S^{2n+2}_{\mathrm T}(A_0,B_0; A_1,B_1)$ defined in \eqref{stre-ac} is invariant under the infinitesimal gauge transformations \eqref{eq:gauge-tra-comp}:
		\begin{equation}\label{stars}
			\delta_{a,b} S^{2n+2}_{\mathrm T}(A_0,B_0;A_1,B_1)=0.
		\end{equation}
	\end{proposition}
	\begin{proof}
		Under the transformation \eqref{eq:gauge-tra-comp}, the fields $A_i$ and $B_i$ ($i=0,1$) transform as
		\begin{align}\label{81}
			\delta_{a,b}A_i = da + \mu_2(A_i,a) - \mu_1(b), \qquad
			\delta_{a,b}B_i = db + \mu_2(A_i,b) - \mu_2(a,B_i) + \frac12\mu_3(a,A_i,A_i).
		\end{align}
		By the variational formula \eqref{eq-TT}, it suffices for proving \eqref{stars} to show that the integrand
		\begin{equation}\label{eq:I-expand}
			n\langle \delta_{a,b}A, F^{n-1};H\rangle + \langle F^n;\delta_{a,b}B\rangle
		\end{equation}
		is an exact differential form. Substituting \eqref{81} into \eqref{eq:I-expand} and expanding gives
		\begin{align}\label{eq:I-expandd}
			n\langle da+\mu_2(A,a)-\mu_1(b),F^{n-1};H\rangle 
			+ \langle F^n;db+\mu_2(A,b)-\mu_2(a,B)+\tfrac12\mu_3(a,A,A)\rangle.
		\end{align}
		
		Combining \eqref{eq:d-pairing} with the 2-Bianchi identities \eqref{2-bianchi}, the terms in \eqref{eq:I-expandd} that involve $da$ and $db$ reduce to
		\begin{align}
			\langle da,F^{n-1};H\rangle
			=\;&d\langle a,F^{n-1};H\rangle
			-\langle \mu_2(A,a),F^{n-1};H\rangle -\big\langle a,F^{n-1};\mu_2(F,B)-\frac12\mu_3(A,A,F)\big\rangle,\label{eq:q0}\\
			\langle F^n;db\rangle
			=\;&d\langle F^n;b\rangle
			-n\langle \mu_1(H),F^{n-1};b\rangle
			-\langle F^n;\mu_2(A,b)\rangle.\label{eq:q1}
		\end{align}
		Furthermore, the invariance identities \eqref{eq:invP-mu2-forms} and \eqref{eq:invP-mu3-forms} imply the following relations
		\begin{align}
			\langle a,F^{n-1};\mu_2(F,B)\rangle
			&=-\langle \mu_2(F,a),F^{n-1};B\rangle,\label{eq:q2}\\
			\langle F^n;\mu_2(a,B)\rangle
			&=n\langle \mu_2(F,a),F^{n-1};B\rangle,\label{eq:q3}\\
			n\langle a,F^{n-1};\mu_3(A,A,F)\rangle
			&=-\langle F^n;\mu_3(A,A,a)\rangle.\label{eq:q4}
		\end{align}
		Substituting \eqref{eq:q0}--\eqref{eq:q4} into \eqref{eq:I-expandd}, we obtain
		\begin{equation*}
			n \,\langle \delta_{a,b}A, F^{\,n-1}; H\rangle
			+ \langle F^{\,n}; \delta_{a,b}B\rangle =nd\langle a, F^{n-1};H\rangle +d\langle F^n; b\rangle.
		\end{equation*}
		
		Finally, applying this identity to both $(A_0,B_0)$ and $(A_1,B_1)$, integrating over $\partial M$, and using the variational formula \eqref{eq-TT} yields
		\begin{align*}
			\delta_{a,b}S^{2n+2}_{\mathrm T}(A_0,B_0;A_1,B_1)
			=\int_{\partial M}\Big(
			n\langle a,F_1^{n-1};H_1\rangle+\langle F_1^n;b\rangle
			-n\langle a,F_0^{n-1};H_0\rangle-\langle F_0^n;b\rangle
			\Big)=0.
		\end{align*}
		Thus, the transgression action is invariant under infinitesimal gauge transformations.
	\end{proof}

	Consequently, within the framework of 2-term 
	$L_\infty$-algebras, the higher transgression action is invariant under infinitesimal gauge transformations on manifolds without boundary. In particular, the strict Lie 
	2-algebra transgression gauge theory~\cite{DHS-4} emerges as a special case and therefore inherits the same infinitesimal gauge invariance. Since HCS theories are themselves concrete realizations of higher transgression gauge theories, their infinitesimal gauge invariance on closed manifolds follows immediately. We omit the details here.

	\section{The Extended Cartan Homotopy Formula in $2n+2$ Dimensions}\label{6-s}
	In ordinary and generalized Chern--Simons theories, the ECHF provides a unified derivation of the Chern--Weil theorem and the triangle equation~\cite{FEP,FIERPS,FIPSPSS}, and clarifies the relation between transgression forms and CS forms~\cite{PRRJ}. In this section, we show that the constructions and results for the ECHF can be generalized to the case in which the invariant form is  the $(2n+3)$-dimensional higher Pontryagin--Chern form based on a 2-term $L_{\infty}$-algebra. We shall briefly state the relevant results and give complete proofs for the analogous case of $(2n+2)$-dimensional strict HCS theory~\cite{DHS-5}.

	Let $\{(A_i,B_i)\}_{i=0}^{r+1}$ be a collection of $\mathfrak{L}$-valued $2$-connections, where $\mathfrak{L}$ is a $2$-term $L_\infty$-algebra.  Let $\triangle_{r+1}$ denote the
	oriented $(r+1)$-simplex with coordinates $\{t^i\}_{i=0}^{r+1}$,  satisfying
	\begin{equation*}
		\sum_{i=0}^{r+1} t^i = 1,
		\qquad
		t^i\ge 0,
		\qquad i=0,\dots,r+1.
	\end{equation*}
	We introduce the interpolating fields
	\begin{equation*}
		A_t \coloneqq  \sum_{i=0}^{r+1} t^i A_i, 
		\qquad 
		B_t \coloneqq  \sum_{i=0}^{r+1} t^i B_i,
	\end{equation*}
	which define a family of $2$-connection  parametrized by $\triangle_{r+1}$. 
	The corresponding curvatures are
	\begin{align}\label{eq:FtHt}
		F_t=dA_t+\frac12\,\mu_2(A_t,A_t)+\mu_1(B_t),
		\qquad
		H_t=dB_t+\mu_2(A_t,B_t)-\frac16\,\mu_3(A_t,A_t,A_t).
	\end{align}
	Notice that, in the semistrict case, the interpolant need not transform as a $2$-connection under the gauge transformations, owing to the nonlinear $\mu_3$-term in the transformation of $B$.  
	
	We can picture each $(A_i,B_i)$ as associated to the $i$-th vertex of $\triangle_{r+1}$. For example for $r=1$, the simplex $\triangle_{2}$ and its oriented boundary orientation are represented by
	\begin{center}
		$\triangle_{2}=$  
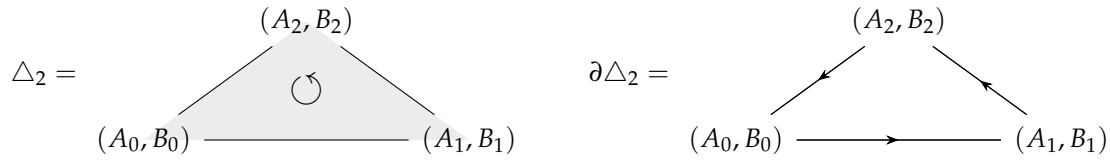
		\begin{tikzcd}[
			row sep=2.6em,
			column sep=1.6em,
			arrows={-},
			execute at end picture={
				\begin{pgfonlayer}{background}
					\fill[gray!15] (A.center) -- (B.center) -- (C.center) -- cycle;
				\end{pgfonlayer}
			}
			]
			& |[alias=C]| (A_2, B_2) \arrow[dr] \arrow[dl] & \\
			|[alias=A]| (A_0, B_0) \arrow[rr] & & |[alias=B]| (A_1, B_1)
			\arrow[from=2-1, to=2-3, phantom,
			pos=.5, overlay,
			"\raisebox{3.9em}{\Large$\circlearrowleft$}"]
		\end{tikzcd}
		\qquad
		$\partial \triangle_{2}=$  
		\begin{tikzcd}[row sep=2.6em, column sep=1.6em]
			& (A_2, B_2) \arrow[dl, no head] & \\
			(A_0, B_0) \arrow[rr, no head] & & (A_1, B_1) \arrow[ul, no head]
			\arrow[no head, from=1-2, to=2-1, draw, overlay,
			decoration={markings, mark=at position .5 with {\arrow{Stealth}}},
			postaction=decorate]
			\arrow[no head, from=2-1, to=2-3, draw, overlay,
			decoration={markings, mark=at position .5 with {\arrow{Stealth}}},
			postaction=decorate]
			\arrow[no head, from=2-3, to=1-2, draw, overlay,
			decoration={markings, mark=at position .5 with {\arrow{Stealth}}},
			postaction=decorate]
		\end{tikzcd}
	\end{center}
	where the arrows on $\partial \triangle_{2}$ indicate the induced
	orientation.
	We accordingly denote the simplex as
	\begin{equation*}
		\triangle_{r+1}=(A_0,B_0;\dots;A_{r+1},B_{r+1})
	\end{equation*}
	and its oriented boundary as
	\begin{equation*}
		\partial \triangle_{r+1}
		=\sum_{i=0}^{r+1}(-1)^i\,\triangle_r^{(i)}(A_0,B_0;\dots;\widehat{A_i},\widehat{B_i};\dots;A_{r+1},B_{r+1}),
	\end{equation*}
	where the hat denotes omission of the corresponding pair.
	
	With the preceding notation, the ECHF is given by~\cite{Zumino}
	\begin{equation}\label{eq:ECHF-int}
		\int_{\partial \triangle_{r+1}}\frac{l_t^{\,p}}{p!}\,\Pi
		=
		\int_{\triangle_{r+1}}\frac{l_t^{\,p+1}}{(p+1)!}\,d\Pi
		+(-1)^{r}\,
		d\int_{\triangle_{r+1}}\frac{l_t^{\,p+1}}{(p+1)!}\,\Pi .
	\end{equation}
	Here $\Pi$ denotes a polynomial in the forms  $\{A_t,B_t,F_t,H_t,d_tA_t,d_tB_t,d_tF_t,d_tH_t\}$ which is also an $m$-form on $M$ and a $q$-form on $\triangle_{r+1}$, with $m\ge p$ and $p+q=r$.
	We denote by $d$ and $d_t$ the exterior derivatives on $M$ and $\triangle_{r+1}$, respectively. The homotopy derivation $l_t$ is defined by the mapping
	\begin{equation*}
		l_t:\Omega^a(M)\times \Omega^b(\triangle_{r+1})\longrightarrow \Omega^{a-1}(M)\times \Omega^{b+1}(\triangle_{r+1})
	\end{equation*}
	and it satisfies Leibniz's rule. Together with $d$ and $d_t$, these operators satisfy the graded algebra relations
	\begin{equation}\label{eq:graded-alg}
		d^2=d_t^2=0,\qquad
		d\,d_t+d_t\,d=0,\qquad
		l_t d-d\,l_t=d_t,\qquad
		l_t d_t=d_t l_t.
	\end{equation}

	Following \cite{FIPSPSS, DHS-5}, we define the action of  $l_t$ on the
	interpolating fields by
	\begin{equation*}
		l_tA_t=0,
		\qquad
		l_tB_t=0,
	\end{equation*}
	which together with \eqref{eq:FtHt} and \eqref{eq:graded-alg}  implies
	\begin{equation}\label{eq:lt-on-curvatures}
		l_tF_t=d_tA_t=dt\dfrac{\partial A_t}{\partial t},
		\qquad
		l_tH_t=d_tB_t=dt\dfrac{\partial B_t}{\partial t}.
	\end{equation}
	
	We now specialize to the particular case of \eqref{eq:ECHF-int} in which the invariant polynomial is given by
	\begin{equation*}
		\Pi\coloneqq \big\langle F_t^{\,n},\,H_t\big\rangle,
	\end{equation*}
	which is closed (that is, $d\Pi=0$) by Proposition \ref{prop:PC-closed}.
	Moreover, $\Pi$ is a 0-form on $\triangle_{p+1}$, i.e., $q=0$, and $\Pi$ is a $(2n+3)$-form on $M$, i.e., $m=2n+3$. Hence, setting $r=p$ in \eqref{eq:ECHF-int}, we obtain
	\begin{equation}\label{eq:higher-descent}
		\int_{\partial \triangle_{p+1}}\frac{l_t^{\,p}}{p!}\,
		\big\langle F_t^{\,n};H_t\big\rangle
		=
		(-1)^p\,d\int_{\triangle_{p+1}}\frac{l_t^{\,p+1}}{(p+1)!}\,
		\big\langle F_t^{\,n};H_t\big\rangle.
	\end{equation}
	Formally, the degree constraint allows $p=0, \dots, 2n+3$.
	These identities constitute the complete set of descent equations for the HCS gauge theory associated with $2$-term $L_\infty$-algebras. 
	They provide the common structural origin of the higher Chern--Weil
	theorem and the triangle equation, which will be derived in the next
	subsection.

	\subsection[Case p=0:  higher Chern--Weil theorem] {Case $p=0$:  higher Chern--Weil theorem}\label{hECHF1}
	We first consider the case $p=0$  of the descent equation \eqref{eq:higher-descent}. For the oriented $1$-simplex $\triangle_1=(A_0,B_0;A_1,B_1)$,
	the descent relation reads
	\begin{equation}\label{dec1}
		\int_{\partial \triangle_1}\langle F_t^{\,n};H_t\rangle
		=
		d\int_{\triangle_1} l_t\langle F_t^{\,n};H_t\rangle.
	\end{equation}
	Here $(F_t,H_t)$ denotes the $2$-curvature associated with the
	interpolation
	\begin{equation*}
		A_t=(1-t)A_0+tA_1,
		\qquad
		B_t=(1-t)B_0+tB_1,
		\qquad
		0\leq t\leq1.
	\end{equation*}
	The induced orientation of the boundary is
	$\partial\triangle_1=(A_1,B_1)-(A_0,B_0)$.
	Therefore, we have
	\begin{equation}\label{dec1-l}
		\int_{\partial\triangle_1}
		\langle F_t^{\,n};H_t\rangle
		=
		\langle F_1^{\,n};H_1\rangle
		-
		\langle F_0^{\,n};H_0\rangle.
	\end{equation}
	
	By defining the difference fields
	$\Theta_{01}= A_1-A_0$ and
	$\Phi_{01}= B_1-B_0$, and applying \eqref{eq:lt-on-curvatures}, we obtain
	\begin{equation*}
		l_tF_t=dt\,\Theta_{01},
		\qquad
		l_tH_t=dt\,\Phi_{01}.
	\end{equation*}
	On the other hand, the symmetric nature of $\langle F_t^{\,n};H_t\rangle$ implies that
	\begin{equation}\label{dec1-r}
		l_t\langle F_t^{\,n};H_t\rangle
		=
		n\langle l_tF_t,F_t^{\,n-1};H_t\rangle
		+
		\langle F_t^{\,n};l_tH_t\rangle.
	\end{equation}
	
	By \eqref{dec1} and \eqref{dec1-l}, integrating \eqref{dec1-r} over the 1-simplex $\triangle_1$ immediately yields the higher Chern--Weil transgression formula (Theorem~\ref{HCW-t}):
	\begin{equation*}
		\langle F_1^{\,n};H_1\rangle
		-
		\langle F_0^{\,n};H_0\rangle
		=
		d\,T^{2n+2}(A_0,B_0;A_1,B_1),
	\end{equation*}
	where the transgression form is explicitly given by
	\begin{align}\label{trans-1}
		T^{2n+2}(A_0,B_0;A_1,B_1)
		\coloneqq
		\int_{\triangle_1}
		l_t\langle F_t^{\,n};H_t\rangle
		=
		\int_0^1dt\,
		\Big( 
		n\langle \Theta_{01},F_t^{\,n-1};H_t\rangle
		+
		\langle F_t^{\,n};\Phi_{01}\rangle
		\Big).
	\end{align}
	This concludes our derivation of the higher Chern--Weil theorem as a corollary of the ECHF.

	\subsection[Case p=1:  higher triangle equation] {Case $p=1$:  higher triangle equation}\label{hECHF2}
	We next consider the case
	$p=1$  of  the descent equation \eqref{eq:higher-descent}. For the oriented 2-simplex
	$\triangle_2=(A_0,B_0;\,A_1,B_1;\,A_2,B_2)$,
	the descent relation becomes
	\begin{equation}\label{eq:hte}
		\int_{\partial \triangle_{2}} l_t \big\langle F_t^{\,n}; H_t\big\rangle
		=
		-\,d \int_{\triangle_{2}} \frac{l_t^{2}}{2}\,\big\langle F_t^{\,n}; H_t\big\rangle.
	\end{equation}
	The interpolating fields are given by
	\begin{align*}
		A_t=t^0A_0+t^1A_1+t^2A_2,\qquad
		B_t=t^0B_0+t^1B_1+t^2B_2,
	\end{align*}
	where $t^0+t^1+t^2=1$  ($t^i\geq 0$, $i=0, 1, 2$).
	
	In this case, the oriented boundary of $\triangle_2$ is
	$\partial \triangle_2 = \triangle^{(0)}_1 - \triangle^{(1)}_1 + \triangle^{(2)}_1$,
	where
	\begin{equation*}
		\triangle^{(0)}_1=(A_1,B_1;\,A_2,B_2),\qquad
		\triangle^{(1)}_1=(A_0,B_0;\,A_2,B_2),\qquad
		\triangle^{(2)}_1=(A_0,B_0;\,A_1,B_1).
	\end{equation*}
	Thus, by \eqref{trans-1}, we get
	\begin{align}\label{eq:LHS-triangle}
		\int_{\partial \triangle_{2}} l_t \big\langle F_t^{\,n}; H_t\big\rangle
		=T^{2n+2}(A_1,B_1;\,A_2,B_2)
		- T^{2n+2}(A_0,B_0;\,A_2,B_2)+ T^{2n+2}(A_0,B_0;\,A_1,B_1).
	\end{align}
	
	To evaluate the right-hand side of \eqref{eq:hte}, we further introduce the new
	coordinates
	\begin{equation*}
		t'\coloneqq 1-t^0,
		\qquad
		s\coloneqq t^2.
	\end{equation*}
	These satisfy
	\begin{equation*}
		t^0=1-t',
		\qquad
		t^1=t'-s,
		\qquad
		t^2=s,
	\end{equation*}
	and transform the simplex into the region $0\leq s\leq t'\leq 1$.
	Using these coordinates, we can write the interpolating fields as
	\begin{align*}
		A_{t}
		&=
		A_0+t'\Theta_{01}+s\Theta_{12},\qquad \Theta_{01}=A_1-A_0, \qquad \Theta_{12}=A_2-A_1,
		\nonumber\\
		B_{t}
		&=
		B_0+t'\Phi_{01}+s\Phi_{12}, \qquad \Phi_{01}=B_1-B_0, \qquad \Phi_{12}=B_2-B_1.
	\end{align*}
	By  \eqref{eq:lt-on-curvatures}, we obtain
	\begin{equation*}
		l_tF_{t}
		=
		dt'\,\Theta_{01}+ds\,\Theta_{12},
		\qquad
		l_tH_{t}
		=
		dt'\,\Phi_{01}+ds\,\Phi_{12}.
	\end{equation*}
	
	On the other hand,  Leibniz's rule for $l_t$ yields
	\begin{align*}
		\frac{l_t^2}{2!}\langle F_{t}^{\,n};H_{t}\rangle
		=
		\frac{n(n-1)}{2}
		\langle (l_tF_{t})^2,F_{t}^{\,n-2};H_{t}\rangle
		-
		n\langle l_tF_{t},F_{t}^{\,n-1};l_tH_{t}\rangle.
	\end{align*}
	By integrating the above identity over the 2-simplex $\triangle_{2}$, we get
	\begin{align}\label{eq:def-secondary-transgression}
		T^{2n+1}(A_0,B_0;A_1,B_1;A_2,B_2)
		\coloneqq&
		\int_{\triangle_2}
		\frac{l_t^2}{2!}\langle F_t^{\,n};H_t\rangle\nonumber\\
		=&
		\int_0^1dt'\int_0^{t'} ds\,
		\Big(
		n(n-1)
		\langle \Theta_{12},\Theta_{01},
		F_{t}^{\,n-2};H_{t}\rangle\nonumber\\
		&
		-n\langle \Theta_{01},F_{t}^{\,n-1};\Phi_{12}\rangle
		+n\langle \Theta_{12},F_{t}^{\,n-1};\Phi_{01}\rangle
		\Big),
	\end{align}
	which is a $(2n+1)$-form on $M$ associated with the oriented $2$-simplex $\triangle_2$.
	
	Combining \eqref{eq:hte}, \eqref{eq:LHS-triangle}, and
	\eqref{eq:def-secondary-transgression}, we obtain the higher triangle
	equation
	\begin{equation}\label{eq:higher-triangle}
		\begin{aligned}
			T^{2n+2}(A_0,B_0;A_2,B_2)
			=&
			T^{2n+2}(A_1,B_1;A_2,B_2)
			+
			T^{2n+2}(A_0,B_0;A_1,B_1)\\
			&+
			dT^{2n+1}(A_0,B_0;A_1,B_1;A_2,B_2).
		\end{aligned}
	\end{equation}
	Furthermore, this triangle identity \eqref{eq:higher-triangle} suggests a recursive decomposition of higher transgression forms, analogous to the subspace-separation procedure in ordinary CS theory
	\cite{FIERPS}. In the present work, we do not provide a detailed implementation of this construction in the semistrict setting, leaving it for future research.
	We emphasize here that the ECHF has enabled us to determine the precise form of the boundary contribution $T^{2n+1}(A_0,B_0;A_1,B_1;A_2,B_2)$ given by \eqref{eq:def-secondary-transgression}. This further reveals that the ECHF is the common origin of both the higher Chern--Weil theorem and the higher triangle equation.
	
	Finally,  specializing $(A_1, B_1)=(0, 0)$ in
	\eqref{eq:higher-triangle} and using \eqref{T-HCS} yields
	\begin{align*}
		T^{2n+2}(A_0,B_0;\,A_2,B_2)
		&= \mathfrak{C}^{2n+2}(A_2,B_2) - \mathfrak{C}^{2n+2}(A_0,B_0) 
		+ d\,T^{2n+1}(A_0,B_0;\,0,0;\,A_2,B_2).
	\end{align*}
	This shows that the higher transgression form is the difference of two HCS forms up to an exact term. The same relation can also be obtained directly by applying the Cartan homotopy formula to the HCS forms. The argument closely follows the one used in strict HCS theory \cite{DHS-5}, with the main difference being that the corresponding construction is based on semistrict $2$-term $L_\infty$-algebras.

	\section{Conclusion and outlook}\label{7-s}
	In this work we construct a semistrict HCS theory associated with balanced $2$-term $L_\infty$-algebras. Working within the homotopy MC formalism, we generalize  the four-dimensional  HCS theory to $2n+2$ dimensions. More specifically, we have established the following results:
	\begin{itemize}
		\item[(i)] Using \eqref{ei}, we construct the higher Pontryagin--Chern form associated with balanced $2$-term $L_\infty$-algebras. We prove that this form is closed and gauge invariant, and derive from it an explicit formula for the $(2n+2)$-dimensional HCS form.
		\item[(ii)] We generalize the Chern--Weil theorem to semistrict $2$-term $L_\infty$-algebras by means of a transgression construction that parallels the strict Lie $2$-algebra case. The associated higher transgression form admits an explicit integral expression, obtained by extending the linear-homotopy method of \cite{DHS-4,FIERPS}, and reduces to the HCS form for a trivial reference $2$-connection.
		\item[(iii)] Varying the higher transgression action reveals that two independent HCS theories defined on the manifold $M$ become inextricably coupled at the boundary, in analogy with ordinary CS theory and its strict higher counterparts. Moreover, when $M$ has no boundary, the higher transgression action is invariant under infinitesimal gauge transformations.
		\item[(iv)] We establish that the ECHF is compatible with the semistrict structure and provides a unified derivation of both the higher Chern--Weil theorem and the higher triangle equation. In the special case where the intermediate fields vanish, i.e., $(A_1,B_1)=(0,0)$, the higher triangle equation implies that the higher transgression form differs from the difference of the two corresponding HCS forms by an exact term.
	\end{itemize}

	Several natural directions remain open for future investigation. It is natural to ask whether the homotopy MC framework used in this paper admits an extension to $3$-gauge theories described by $2$-crossed modules and related models of strict higher gauge theory \cite{Faria_Martins_2011,doi:10.1063/1.4870640,Song,TRMV1,Radenkovic:2019qme}. Such a generalization would require the construction of appropriate higher invariant polynomials and the corresponding invariant forms for the associated $3$-term $L_\infty$-algebraic data. 
	Another question is whether the five-dimensional HCS theory constructed by Song et al. in the strict setting \cite{DHS-JHEP} possesses higher-dimensional analogues within the semistrict homotopy-theoretic framework introduced here. We leave these questions for future work.
	Although the significance of higher-dimensional HCS theory within mathematical physics remains to be fully understood, we hope that the present work makes some first steps toward its development.

	\section*{Acknowledgements}
	Danhua Song was supported by the China Postdoctoral Science Foundation (Grant No. 2026M793335).


\begin{thebibliography}{99}
	\bibitem{RZ-AKSZ}
	Zucchini, R.: AKSZ models of semistrict higher gauge theory. J. High Energy Phys. \textbf{2013}(3), 14 (2013)
	
	\bibitem{Fiorenza-Rogers-Schreiber}
	Fiorenza, D., Rogers, C.L., Schreiber,  U.: A higher Chern--Weil derivation of AKSZ $\sigma$-models. Int. J. Geom. Methods Mod. Phys. \textbf{10}(1), 1250078 (2013) 
	
	\bibitem{PRCS-2016}
	Ritter, P., Saemann, C.: $L_{\infty}$-algebra models and higher Chern--Simons theories. Rev. Math. Phys. \textbf{28}(9), 1650021 (2016)
	
	\bibitem{GR}
	Gagliardo, G., Rist, D., Saemann, C., Wolf, M.:	Adjusting higher Chern--Simons theory. \href{https://doi.org/10.48550/arXiv.2507.02082}{arXiv:2507.02082} (2025)
	
	\bibitem{BLCM}Jur\v{c}o, B., Raspollini, L., Saemann, C., Wolf, M.: $L_{\infty}$-algebras of classical field theories and the Batalin--Vilkovisky formalism. Fortsch. Phys. \textbf{67}(7), 1900025 (2019) 
	
	\bibitem{DHS-4} Song, D.H., Wu, K., Yang, J.: Higher Chern--Simons--Antoniadis--Savvidy forms based on crossed modules. Phys. Lett. B \textbf{848}, 138374 (2024) 
	
	\bibitem{DHS-5}Song, D.H.: Extended Cartan homotopy formula for higher Chern--Simons--Antoniadis--Savvidy theory. Phys. Lett. B \textbf{865}, 139471 (2025)
	
	
	\bibitem{IAGS}
	Antoniadis, I., Savvidy, G.: New gauge anomalies and topological invariants invarious dimensions. Eur. Phys. J. C \textbf{72}(9), 2140 (2012)
	
	\bibitem{FIPSPSS}
	Izaurieta, F., Salgado,  P., Salgado, S.: Chern--Simons--Antoniadis--Savvidy forms and standard supergravity. Phys. Lett. B \textbf{767}, 360--365 (2017)
	
	
	\bibitem{FIP} Izaurieta, F., Mu\~noz, I., Salgado, P.: A Chern--Simons gravity action in $d=4$. Phys. Lett. B \textbf{750}, 39--44 (2015)
	
	
	\bibitem{Girelli-Pfeiffer}
	Girelli, F., Pfeiffer, H.: Higher gauge theory-differential versus integral formulation. J. Math. Phys. \textbf{45}(10), 3949--3971 (2004)
	
	\bibitem{Baez-2007}
	Baez, J.C., Schreiber, U.: Higher gauge theory. In: Davydov, A. et al. (eds.) Categories in Algebra, Geometry and Mathematical Physics, pp. 7--30. AMS, Providence, Rhode Island (2007)%
	
	\bibitem{Baez.2010}
	Baez, J.C., Huerta, J.: An invitation to higher gauge theory. Gen. Relativity Gravitation \textbf{43}(9), 2335--2392 (2011) 
	
	\bibitem{Cagnacci-19}
	Cagnacci, Y., Codina, T., Marques, D.: $L_{\infty}$ algebras and tensor hierarchies in exceptional field theory and gauged supergravity. J. High Energy Phys. \textbf{2019}(1), 117 (2019)
	
	\bibitem{Zeitlin-09}
	Zeitlin, A.M.: String field theory-inspired algebraic structures in gauge theories. J. Math. Phys. \textbf{50}(6), 063501 (2009)
	
	\bibitem{Polchinski}
	Polchinski, J.: String theory. Vol. II. Cambridge Monographs on Mathematical Physics, Cambridge Univ. Press, Cambridge (1998)%
	
	\bibitem{Becker-2007}
	Becker, K., Becker, M., Schwarz, J.H.: String theory and M-theory. Cambridge Univ. Press, Cambridge (2007)%
	
	\bibitem{Baez-543}
	Baez, J.C.: An introduction to spin foam models of  BF theory and quantum gravity. In:  Gausterer, H., Grosse, H., Pittner, L.  (eds.) Geometry and quantum physics, pp. 25--93. Springer, Berlin (1999)
	
	\bibitem{CR-455}
	Rovelli, C.: Quantum gravity. Cambridge Monographs on Mathematical Physics, Cambridge Univ. Press, Cambridge (2004)%
	
	\bibitem{MSJS}
	Schlessinger, M., Stashe, J.: The Lie algebra structure of tangent cohomology and deformation theory. J. Pure Appl. Algebra \textbf{38}, 313--322 (1985) 
	
	\bibitem{TLJS}
	Lada, T., Stashe, J.: Introduction to SH Lie algebras for physicists.  Internat. J. Theoret. Phys.  \textbf{32}(7), 1087--1103 (1993)
	
	\bibitem{MP}
	Penkava, M.: $L_{\infty}$-algebras and their cohomology. \href{https://doi.org/10.48550/arXiv.q-alg/9512014}{arXiv:q-alg/9512014} (1995)
	
	\bibitem{Kraft-Schnitzer}
	Kraft, A., Schnitzer, J.: An introduction to $L_\infty $-algebras and their homotopy theory for the working mathematician. Rev. Math. Phys. \textbf{36}(1), 2330006 (2024) 
	
	\bibitem{B. Zwiebach-1993}
	Zwiebach, B.: Closed string field theory: quantum action and the BV master equation.
	Nuclear Phys. B \textbf{390}, 33--152 (1993)
	
	\bibitem{JCB-ASC}
	Baez, J.C., Crans, A.S.: Higher-dimensional algebra VI: Lie 2-algebras. Theor. Appl. Categor. \textbf{12},  492--528 (2004)
	
	\bibitem{Zucchini-2014-1}
	Soncini, E., Zucchini, R.: 4-D semistrict higher Chern--Simons theory I. J. High Energy Phys. \textbf{2014}(10), 79 (2014)
	
	\bibitem{Zucchini-2016}
	Zucchini, R.: A Lie based 4-dimensional higher Chern--Simons theory. J. Math. Phys. \textbf{57}(5),  052301 (2016)
	
	\bibitem{HC-2024}
	Chen, H., Liniado, J.: Higher gauge theory and integrability. Phys. Rev. D \textbf{110}(8), 086017  (2024)
	
	\bibitem{Zucchini-2021}
	Zucchini, R.: 4-d Chern--Simons theory: higher gauge symmetry and holographic aspects. J. High Energy Phys. \textbf{2021}(6), 25 (2021) 
	
	\bibitem{Zucchini-I} 
	Zucchini, R.: On higher holonomy invariants in higher gauge theory I. Int. J. Geom. Methods Mod. Phys. \textbf{13}(7), 1650090 (2016)  
	
	\bibitem{Zucchini-II}
	Zucchini, R.: On higher holonomy invariants in higher gauge theory II. Int. J. Geom. Methods Mod. Phys. \textbf{13}(7), 1650091 (2016)  
	
	\bibitem{Zucchini-2019}  
	Zucchini, R.: Wilson Surfaces for Surface Knots: a field theoretic route to higher knots. Fortschr. Phys.  \textbf{67}(8-9), 1910026 (2019)
	
	\bibitem{Schenkel-Vicedo}
	Schenkel, A., Vicedo, B.: 5d 2-Chern--Simons theory and 3d integrable field theories. Comm. Math. Phys. \textbf{405}(12),  293 (2024)
	
	\bibitem{DHS-JHEP}
	Song, D.H., Wu, M.Y., Wu, K., Yang, J.: Higher Chern--Simons based on (2-)crossed modules. J. High Energy Phys. \textbf{2023}(7), 207 (2023)
	
	\bibitem{Salgado}
	Salgado, S.: On the $L_{\infty}$ formulation of Chern--Simons theories. J. High Energy Phys. \textbf{2022}(4), 142 (2022)
	
	\bibitem{SS}
	Salgado, S.: Gauge-invariant theories and higher-degree forms. J. High Energy Phys. \textbf{2021(10)}, 66 (2021)
	
	\bibitem{PSSS}
	Salgado, P., Salgado, S.: Extended gauge theory and gauged free differential algebras. Nuclear Phys. B \textbf{926}, 179--199 (2018) 
	
	\bibitem{GS}
	Savvidy, G.: Topological mass generation in four-dimensional gauge theory. Phys. Lett. B \textbf{694(1)}, 65--73 (2010) 
	
	\bibitem{GS1}
	Savvidy, G.: Extension of Chern--Simons forms and new gauge anomalies. Internat. J. Modern Phys. A \textbf{29}(3-4), 1450027 (2014)
	
	\bibitem{SKGS}
	Konitopoulos, S., Savvidy, G.: Extension of Chern--Simons forms. J. Math. Phys. \textbf{55}(6), 062304 (2014) 
	
	\bibitem{FEP}
	Izaurieta, F., Rodr\'{i}guez, E., Salgado, P.: The extended Cartan homotopy formula and a subspace separation method for Chern--Simons theory. Lett. Math. Phys. \textbf{80}(2), 127--138 (2007)
	
	
	\bibitem{Zumino}
	Ma$\tilde{n}$es, J., Stora,  R., Zumino, B.: Algebraic study of chiral anomalies. Comm. Math. Phys. \textbf{102}(1), 157--174 (1985)
	
	\bibitem{PRRJ}
	Mora, P.,  Olea, R., Troncoso, R., Zanelli, J.: Transgression forms and extensions of Chern--Simons gauge theories. J. High Energy Phys.  \textbf{2006}(2), 067 (2006)
	
	\bibitem{BTLCM}
	Jur\v{c}o, B., Macrelli, T., Raspollini, L., Saemann, C., Wolf, M.: $L_{\infty}$-algebras, the BV formalism, and classical fields. Fortschr. Phys. \textbf{67}(8-9), 1910025 (2019)%
	
	\bibitem{ASCFS}
	Cattaneo, A.S.,  Sch\"atz, F.: Introduction to supergeometry. Rev. Math. Phys. \textbf{23}(6), 669--690 (2011)
	
	\bibitem{NW}
	Woodhouse, N.M.J.: Introduction to analytical dynamics. Springer Undergraduate Mathematics Series, Springer, London (2009)%
	
	\bibitem{OZ} 
	Ortolani, F., Zucchini, R.: Higher Chern--Simons gauge theory. PhD Thesis, Universit{\`a} di Bologna (2015)%
	
	\bibitem{Kontsevich-92}
	Kontsevich, M.: Feynman diagrams and low-dimensional topology. In: Joseph, A., Mignot, F., Murat, F., Prum, B., Rentschler, R. (eds.) First European Congress of Mathematics Paris, pp. 97--121. Birkh\"auser, Basel (1994)
	
	\bibitem{BJCSMW} 
	Jurco, B., Saemann, C., Wolf, M.: Higher groupoid bundles, higher spaces, and self-dual tensor field equations. Fortschr. Phys. \textbf{64}(8-9), 674--717 (2016)
	
	\bibitem{WS}
	Wu, M.Y., Song, D.H.: Higher descent equations based on 2-term $L_{\infty}$-algebras. \href{https://doi.org/10.48550/arXiv.2603.27588}{arXiv: 2603.27588} (2026)
	
	\bibitem{GGGS}
	Georgiou, G., Savvidy, G.: Mixed symmetry tensor fields and new topological invariants. \href{https://doi.org/10.48550/arXiv.1212.5228}{arXiv:1212.5228} (2012)
	
	\bibitem{BZ}
	Zumino, B.: Chiral anomalies and differential geometry. In: Sibold, K. (eds) Relativity, groups and topology, II, pp. 1291--1322. North-Holland, Amsterdam (1984)
	
	\bibitem{FIERPS}
	Izaurieta, F., Rodr\'{i}guez,  E., Salgado, P.: On transgression forms and Chern--Simons (super)gravity. \href{https://doi.org/10.48550/arXiv.hep-th/0512014}{arXiv:hep-th/0512014} (2005)
	
	
	\bibitem{Faria_Martins_2011} 
	Martins, J.F., Picken, R.: The fundamental Gray 3-groupoid of a smooth manifold and local 3-dimensional holonomy based on a 2-crossed module.  Differential Geom. Appl.  \textbf{29}(2), 179--206 (2011)
	
	\bibitem{doi:10.1063/1.4870640}  
	Wang, W.: J. Math. Phys. On 3-gauge transformations, 3-curvatures, and gray-categories. J. Math. Phys. \textbf{55}(4), 043506 (2014)
	
	\bibitem{Song}
	Song, D.H.,  Lou, K., Wu, K., Yang, J., Zhang, F.H.: 3-form Yang--Mills based on 2-crossed modules. J. Geom. Phys. \textbf{178}, 104537 (2022)
	
	\bibitem{TRMV1}  
	Radenkovi{\'c}, T., Vojinovi{\'c}, M.: Gauge symmetry of the $3$BF theory for a generic Lie three-group. Classical Quantum Gravity \textbf{39}(13), 135009 (2022)
	
	\bibitem{Radenkovic:2019qme}  
	Radenkovi{\'c}, T., Vojinovi{\'c}, M.: Higher gauge theories based on 3-groups. J. High Energy Phys. \textbf{2019}(10), 222  (2019) 
	
	
	
	
	
	
	
	
	
	
	
	
	
	
	
	
	
	
	
	
	
	
	
\end{thebibliography}
\end{document}